\documentclass{amsart}

\usepackage{todonotes}

\title[Notes on the BKP hierarchy]
{From the $B$-Toda to the BKP hierarchy}

\author{Yuancheng Xie} 
\date{\today}
\address{Beijing International Center for Mathematical Research, Peking University,
Beijing 100871, China}
\email{xieyuancheng@bicmr.pku.edu.cn}
\subjclass[2000]{}

\usepackage{multirow}
\usepackage{url}
\usepackage{color}
\usepackage{rotating}

\countdef\x=23
\countdef\y=24
\countdef\z=25
\countdef\t=26

\def\tbox(#1,#2)#3{
\x=#1 \y=#2 
\multiply\x by 12 
\multiply\y by 12 
\z=\x \t=\y
\advance\z by 12 
\advance\t by 12 
\psline(\x,\y)(\x,\t)(\z,\t)(\z,\y)(\x,\y)
\advance\x by 6
\advance\y by 6 
\rput(\x,\y){{\bf #3}}}

\usepackage{amssymb}
\usepackage{graphicx}
\usepackage{epstopdf}
\usepackage{amscd}
\usepackage{appendix}
\usepackage{graphics}
\usepackage{MnSymbol}
\usepackage{tikz-cd}

\usepackage{geometry}
\usepackage{hyperref}

\def\proof{\par{\it Proof}. \ignorespaces}
\def\endproof{{\ \vbox{\hrule\hbox{%
     \vrule height1.3ex\hskip0.8ex\vrule}\hrule }}\par}

\def\sproof{\par{\it Sketch of Proof}. \ignorespaces}
\def\endsproof{{\ \vbox{\hrule\hbox{%
				\vrule height1.3ex\hskip0.8ex\vrule}\hrule }}\par}

\numberwithin{equation}{section}

\let\trueint=\int
\let\truesum=\sum
\def\int{\mathop{\textstyle\trueint}\limits}
\def\sum{\mathop{\textstyle\truesum}\limits}
\let\<=\langle
\let\>=\rangle

\def\t{\mathbf{t}}
\def\0{\mathbf{0}}

\def\edge{\ar@{-}}
\def\dedge{\ar@{.}}

\usepackage{amsmath, amsthm, amssymb, amsbsy,amsfonts,amscd}
\usepackage[mathscr]{eucal}
\usepackage{epsfig}
\newtheorem{theorem}{Theorem}[section]

\newtheorem{definition}[theorem]{Definition}
\newtheorem{proposition}[theorem]{Proposition}
\newtheorem{lemma}[theorem]{Lemma}

\newtheorem{example}[theorem]{Example}
\newtheorem{corollary}[theorem]{Corollary}

\newtheorem{remark}[theorem]{Remark}

\usepackage{amsfonts}
\usepackage{xy}
\usepackage{amssymb}
\usepackage{amsmath}

\newcommand{\ttt}{\mathbf t}

\makeatletter
\renewcommand*\env@matrix[1][*\c@MaxMatrixCols c]{%
  \hskip -\arraycolsep
  \let\@ifnextchar\new@ifnextchar
  \array{#1}}
\makeatother

\newcommand{\thmrefer}[1]{\renewcommand\thetheorem
  {\protect\ref{#1}}\addtocounter{theorem}{-1}}

\xyoption{all}
\CompileMatrices

\begin{document}

\begin{abstract}
It is shown that all $\tau$-functions of BKP hierarchy can be written as Pfaffians of skew-symmetric matrices. $\tau$-functions of BKP hierarchy are parameterized by points in the  universal orthogonal Grassmannian manifold (UOGM). The UOGM is a disjoint union of Schubert cells, we classify and give explicit parameterization for points in each Schubert cell by constructing a frame for UOGM in the sense of Sato. $\tau$-functions are then expressed in terms of these frames and Schur-Q functions. For concreteness we give a comprehensive study for the $\tau$-functions of $B$-Toda which can be viewed as a finite version of the BKP hierarchy. Along the way we also give a constructive description for complex pure spinors du E. Cartan. As an application of our construction, we reprove a theorem due to A. Alexandrov which states that KdV solves BKP up to rescaling of the time parameters by $2$. We prove this by showing that the KdV hierarchy can be viewed as $4$-reduction of the BKP hierarchy.
This interpretation 
gives complete characterization for the KdV orbits inside the BKP hierarchy. Other than a few facts from representation theory, the main tools we use to show the above results, however, are surprisingly simple linear algebra.
\end{abstract}

\maketitle

\tableofcontents

\section{Introduction}
Around forty years ago, inspired by the work of M. Sato and Y. Sato which relates KP hierarchy with the group orbit of $\mathcal{A}_{\infty}$ through the highest weight vector of its fundamental module \cite{Sato1981, Sato-Sato1982}, E. Date, M. Jimbo, M. Kashiwara and T. Miwa introduced in \cite{Date-Jimbo-Kashiwara-Miwa1982} the BKP hierarchy which describes the group orbit of $\mathcal{B}_{\infty}$. A few years later, Y. You studied the rational solutions of the BKP hierarchy in \cite{You1989} and showed that the building blocks of the $\tau$-function of the BKP hierarchy are Schur-Q functions which are associated with characters of projective representations of symmetric groups, and these functions exhibit the Pfaffian structure instead of the determinant structure which is ubiquitous in KP theory. In the same year, Hirota showed in  \cite{Hirota1989} that some soliton solutions of the BKP hierarchy can also be written in the Pfaffian form. BKP hierarchy enjoys the same universality as KP hierarchy in the sense that many other integrable hierarchies can be obtained as reductions of the BKP hierarchy, for example the Sawada-Kotera hierarchy and Ramani hierarchy are $3$- and $5$-reduction of the BKP hierarchy, respectively. However, compared with the rapid development of the KP theory, the study of BKP hierarchy apparently falls behind. Not much progress has been made after its initial blossom in the 1980s until recently there are some renewed interests in BKP hierarchy because of its connection with representation theory and enumerative geometry (c.f. \cite{Kac-Leur2019, Kac-Rozhkovskaya-Leur2021, Mironov-Morozov2021, Alexandrov2021, Liu-Yang2021, Liu-Yang2022, Harnad-Orlov2021a, Harnad-Orlov2021b, Harnad-Orlov2021c}). A major reason accounting for such imbalance is that BKP hierarchy lacks a simple and unified formula for its $\tau$-functions, for example, results in \cite{Shigyo2016} show that when the $\tau$-function of BKP hierarchy is expanded in terms of Schur-Q functions, the coefficients satisfy a set of quite complicated equations and it is not clear where to find proper coefficients verifying all these relations. The main goal of the current paper is to remedy such deficiency in the literature. 

To illustrate what we mean by this, let us first recall M. Sato and Y. Sato's formula of $\tau$-functions for the KP hierarchy.  Let $\Lambda = (\delta_{\mu, \nu-1})_{\mu, \nu \in \mathbb{Z}}$ be the shift operator, $\eta(\t, \Lambda) = \sum \limits_{n = 1}^{\infty}t_n\Lambda^n$ and $\xi_0 = (\delta_{\mu, \nu})_{\stackrel{\mu, \nu \in \mathbb{Z}}{\nu < 0}}, \xi = (\xi_{\mu, \nu})_{\stackrel{\mu, \nu \in \mathbb{Z}}{\nu < 0}}$, then the $\tau$-function of KP hierarchy can be expressed compactly as the determinant of an infinite matrix
\begin{equation}\label{eq:tauKPintro}
\tau_{\text{KP}}(\t; {\xi}) = \tau_{\text{KP}}(t_1, t_2, t_3, \dots; {\xi}) := \det({\xi}^T_0e^{\eta(\t, \Lambda)}{\xi}).
\end{equation}
Note that here $\xi$ parameterizes the solutions, and up to a constant $\tau_{\text{KP}}(\t; {\xi})$ is invariant under the right action of invertible $\mathbb{N} \times \mathbb{N}$ matrices on $\xi_0$ and $\xi$. That is to say KP hierarchy defines a dynamics on a homogeneous space, the so-called universal Grassmannian manifold (UGM), of a properly defined Lie group $\mathcal{A}_{\infty}$ consisting of invertible $\mathbb{Z} \times \mathbb{Z}$ matrices. The elegant formula \eqref{eq:tauKPintro} immediately reveals a lot of interesting and intrinsic structures of the $\tau$-functions, for example, all the Schur polynomials are $\tau$-functions, Schur expansion of $\tau_{\text{KP}}(\t; \xi)$ follows from the Binet-Cauchy formula, and the coefficients of the Schur expansion automatically satisfy all the Pl\"ucker relations, etc.

The parallel theory for BKP hierarchy, however, is a little subtler. First, the solution space of BKP hierarchy are parameterized by the universal orthogonal Grassmannian manifold(UOGM), which are maximal isotropic with respect to the non-degenerate symmetric quadratic form defining the infinite-dimensional type-$B$ Lie group $\mathcal{B}_{\infty}$ and Lie algebra $\mathfrak{b}_{\infty}$, that is to say there exist some serious constraints on the corresponding $\xi$ (see Equation \eqref{eq:tauKPintro}) in type-$B$ theory. Second, the fundamental module we need to consider is from the so-called spin representation of $\mathfrak{b}_{\infty}$ which is unfortunately only familiar to a few experts. Third, even though $\mathfrak{b}_{\infty}$ can be viewed as a sub-Lie algebra of $\mathfrak{a}_{\infty}$,  the building blocks of its $\tau$-function--the Schur Q-function arising in BKP theory does not seem to be directly related to the more familiar Schur function in KP theory at first sight\footnote{Except Schur and Schur-Q functions are both special cases of Hall-Littlewood polynomials.}, let alone the equations satisfied by the coefficients for general $\tau$-functions of KP and BKP hierarchies, respectively. Nevertheless, simple and unified formula for $\tau$-functions of BKP hierarchy do exist. 

As warm up, we can state the following combinatorial formula for $\tau$-functions of BKP hierarchy in the big Schubert cell.
\begin{theorem}\label{thm:taufunctiongeneric}
Let
\[\tilde{J} = \begin{pmatrix}
& & & \udots \\
& & \udots & \\
& 1 & & \\
1 & & & 
\end{pmatrix}, \qquad \text{and} \qquad Q(\t_B) = \begin{pmatrix}
\ddots & \vdots & \vdots & \vdots & \vdots\\
\cdots & 0 & -Q_{3, 2}(\t_B) & -Q_{3, 1}(\t_B) & -Q_{3, 0}(\t_B)\\
\cdots & Q_{3, 2}(\t_B) & 0 & -Q_{2, 1}(\t_B) & -Q_{2, 0}(\t_B)\\
\cdots & Q_{3, 1}(\t_B) & Q_{2, 1}(\t_B) & 0 & -Q_{1, 0}(\t_B)\\
\cdots & Q_{3, 0}(\t_B) & Q_{2, 0}(\t_B) & Q_{1, 0}(\t_B) & 0 
\end{pmatrix},\]
where $Q_{\lambda_i, \lambda_j}$'s are the elementary Schur-$Q$ functions (see Section \ref{sec:SchurQ} for our convention) and $\t_{B} = (t_1, t_3, t_5, \dots)$ is the BKP flow.
Then $\tau$-functions of the BKP hierarchy in the big Schubert cell have the following form:
\begin{equation}\label{eq:Pfaffiantau}
\begin{array}{rcl}
\tau_{\text{BKP}}(\t_B) & = & \text{Pf}
\begin{pmatrix}[ccc|ccc]
 &    &    &  & & \\
 & Q &   &  & \tilde{J} &\\
 &  &     & & & \\
\hline 
 &    & &      & & \\
 & -\tilde{J}^T & &  & R  &\\
 &  &     & & & \\
\end{pmatrix}\\
& = & \sum \limits_{\lambda \in \text{DP}}\text{Pf}(R_{\lambda})Q_{\lambda},
\end{array}
\end{equation}
where $R$ is an arbitrary skew symmetric matrix, i.e. $R = (r_{i, j})_{i, j \ge 0}, r_{i, j} = - r_{j, i}$, Pf denotes the pfaffian of a skew-symmetric matrix and the summation runs over all strict partitions. 
\end{theorem}
First let's make a remark regarding the meaning of the pfaffian for an infinite-dimensional skew-symmetric matrix. Note that the Schur-Q expansion which is most important for our purpose in this paper is consistent with the finite-dimensional truncations, that is we can view this infinite-dimensional skew-symmetric matrix both as injective limit of finite-dimensional $(2n+2) \times (2n+2)$ skew-symmetric matrices or projective limit of finite-dimensional $(2n+2) \times (2n+2)$ skew-symmetric matrices. The difference is that the former gives us an algebraic variety which includes for example all the points representing rational solutions of the BKP hierarchy and the latter contains more general ones such as soliton and quasi-periodic solutions of the BKP hierarchy. The other possibilities for example those with topological structures sit between these two extremes (c.f. \cite{Takasaki1989}).

The Schur-Q expansion for $\tau_{\text{BKP}}(\t_B)$ in the big Schubert cell, i.e. the second expression for $\tau_{\text{BKP}}(\t_B)$ in \eqref{eq:Pfaffiantau} where $\tau(\0)=1$, is essentially known to the Kyoto School  \cite{Date-Jimbo-Kashiwara-Miwa1982} (see also \cite{Harnad-Balogh2021}). The first equality which is missed for so many years actually follows easily from the combinatorial identities in \cite{Ishikawa-Wakayama1995} where we just need to let the size of the relevant matrices go to infinity. We will give a representation theoretical proof for this formula and generalize it to obtain a unified expression for all $\tau_{\text{BKP}}(\t_B)$, including singular solutions of the BKP hierarchy.

Before we state our next result, let us remark that the minor summation formula in \eqref{eq:Pfaffiantau} is really an analogue of Binet-Cauchy formula for determinant in formula \eqref{eq:tauKPintro} in the sense that when we expand the pfaffian expression we need to take the sum of the corresponding product of all possible minors, and our representation theoretical proof reflects exactly this fact (see a similar viewpoint also in \cite{Balogh-Harnad-Hurtubise2021}). 

Formula \eqref{eq:Pfaffiantau} is of course based on the well-known fact that the big Schubert cell in UOGM is parameterized by skew-symmetric matrices $R$ (c.f. \cite{Cartan1966, Chevalley1954, Date-Jimbo-Kashiwara-Miwa1982}). To obtain an expression for general $\tau$-function $\tau_{\text{BKP}}(\t_B)$, we need to parameterize all the Schubert cells in UOGM. Let us introduce some notations in order to state this classification result.  

Let $\tilde{V} \cong \mathbb{C}^{\infty}$ be the infinite-dimensional vector space with basis $(\dots, e_{-n}, \dots, e_{-1}, e_{-0}, e_{+0}, e_{1}, \dots, e_{n}, \dots)$. We equip with $\tilde{V}$ the non-degenerate quadratic form $\tilde{Q}$ such that
\[\tilde{Q}(\sum x_{\alpha}e_{\alpha}) = \sum \limits_{k = 0}^{\infty} (-1)^{k}x_{+k}x_{-k}.\]
Our infinite-dimensional orthogonal Lie algebra $\mathfrak{b}_{\infty}$ is associated with the quadratic subspace $(V, Q) \subset (\tilde{V}, \tilde{Q})$, where $V$ is the orthogonal complement of the vector $(e_{-0} - e_{+0})$ in $\tilde{V}$, and $Q$ is the restriction of $\tilde{Q}$ to $V$. 

The Schubert cells of UOGM are indexed by the normal subgroup $\mathfrak{N}_B$ of the Weyl group $\mathfrak{W}_B$ of $\mathfrak{b}_{\infty}$ which changes the signs of some short roots $\varepsilon_i$'s of $\mathfrak{b}_{\infty}$(see Section \ref{sec:BnDn1}). For each element $w \in \mathfrak{N}_B$, we can associate to it a unique strict partition $\lambda = (\lambda_1, \lambda_{2}, \dots, \lambda_{2\ell})$, where $\lambda_1 > \lambda_2 > \dots > \lambda_{2\ell} \ge 0$ such that $w(\varepsilon_{\lambda_j}) = -\varepsilon_{\lambda_j}$ for all $\lambda_j \ne 0$. We denote this element by $w_{\lambda}$ and the corresponding Schubert cell by $\mathcal{L}_{\lambda}$ in the following. We associate a sequence of increasing integers $(i_1, i_2, \dots, i_{n+1}, \dots)$ to $\lambda$ by the rule that if $i_k \in \{\lambda_1, \lambda_2, \dots, \lambda_{2\ell}\}$ then $i_k = -k$, otherwise $i_k = k$.

Consider a $\mathbb{Z} \times \mathbb{Z}$ matrix $L$. We number the rows and columns of $L$ by 
\[(\cdots, -n, \cdots, -1, -0, +0, 1, \cdots, n, \cdots)\]
 from left to right and from top to bottom. Let $\tilde{E}_{i,j}$ be the matrix with a $1$ at the $i$-th row and $j$-th column and $0$ at other places, then we define
\[\left\{\begin{array}{l}
E_{i, j} = \tilde{E}_{i, j} \qquad \text{for } i, j \not\in \{+0, -{0}\}\\
E_{i,0} = \tilde{E}_{i, +0} + \tilde{E}_{i, -0}\quad \text{and} \quad E_{0, i} = \tilde{E}_{+0, i} + \tilde{E}_{-0, i}.
\end{array}\right.\]

The Schubert cell associated with the strict partition $\lambda$ in $({V}, {Q})$ is parameterized by
\begin{theorem}\label{thm:normalforminf}
Columns $(i_1, i_2, \dots, i_{n+1}, \dots )$ of the following matrix
\[L_{\lambda} := \tilde{M} + \sum \limits_{i_k > |i_l|}r_{i_k, i_l}[E_{i_k, i_l} - E_{i_l, i_k}] - \sum \limits_{i_k, i_l > 0}r_{i_k, 0}r_{i_l, 0}E_{i_k, i_l},\]
give a canonical form for elements in the Schubert cell $\mathcal{L}_{\lambda}$ of UOGM associated with $w_{\lambda} \in \mathfrak{N}_B$.
Here $\tilde{M}$ is the Gram matrix of $\tilde{Q}$ and $r_{i_k, i_l} \in \mathbb{C} \ (i_k > |i_l|)$ are arbitrary constants which provide the inhomogeneous coordinate system for the Schubert cell $\mathcal{L}_{\lambda}$ of UOGM indexed by $w_{\lambda}$. 
\end{theorem}

With the classification result in Theorem \ref{thm:normalforminf}, the general $\tau$-functions of BKP hierarchy are given by

\begin{theorem}\label{thm:taufunctioninf}
$\tau_{\text{BKP}}(\t_B; w_{\lambda})$-function for the BKP hierarchy associated with $L_{\lambda}$ is given by the Pfaffian of the following skew-symmetric matrix
\begin{align*}
W_{\lambda} = & \sum \limits_{j > i \ge 0}Q_{ji}(\t_B) (\tilde{E}_{-i, -j} - \tilde{E}_{-j, -i}) + \sum \limits_{i_k > |i_l|}\left[(\tilde{E}_{-i_k, i_k} - \tilde{E}_{i_k, -i_k}) + (-1)^{\text{min}\{0, i_{l}\}}r_{i_k, i_l}(\tilde{E}_{i_k, i_l} - \tilde{E}_{i_l, i_k})\right] + \\
& \qquad + \sum \limits_{j = 1}^{\ell}(\tilde{E}_{\lambda_{2j}, \lambda_{2j-1}} - \tilde{E}_{\lambda_{2j-1}, \lambda_{2j}}) + (1 - \delta_{\lambda_{2\ell, 0}})(\tilde{E}_{-0, +0} - \tilde{E}_{+0, -0}),
\end{align*}
where $(i_1, i_2, \dots, i_{n+1}, \dots)$ is the index associated with the strict partition $\lambda = (\lambda_1, \lambda_2, \dots, \lambda_{2\ell})$, and $r_{i, j}$'s are constants parameterizing the solutions.
\end{theorem}

Interestingly, Theorem \ref{thm:taufunctioninf} also gives complete characterization of pure spinors of E. Cartan as a byproduct as we now explain. In the main body of the paper, we view type $B$ Toda as finite-dimensional BKP hierarchy and give a comprehensive study for the corresponding finite version $\tau$-functions of the full Kostant-Toda (f-KT) hierarchy associated with the spin representation $(\rho, S)$ of $\mathfrak{so}_{2n+1}$. That is, restricted to the case when dim $V = 2n+1 < \infty$, we have analogues of Theorem \ref{thm:taufunctiongeneric}, Theorem \ref{thm:normalforminf} and Theorem \ref{thm:taufunctioninf} for finite dimensional orthogonal Grassmannian manifolds and polynomial $\tau$-functions of $B$-Toda. 
The underlying vector space $S$ of spin representation is the exterior algebra of the maximal isotropic subspace generated by $\{e_{-1}, e_{-2}, \dots, e_{-n}\}$. That is each element in $S$ can be uniquely written as
\[s = \sum \limits_{n \ge i_1 > i_2 > \cdots > i_k > 0}\xi_{i_1, i_2, \dots, i_k}e_{-i_1} \wedge e_{-i_2} \wedge \cdots \wedge e_{-i_k}.\]
An element $s \in S$ is usually called a spinor, and the coefficients $\xi_{i_1, i_2, \dots, i_k}$ are called the Cartan coordinates of $s$. E. Cartan shows that for each point in the orthogonal Grassmannian manifold $\text{OG}(V)$, there is a way to associate with it an element in $S$ and spinors with such a geometric origin are called pure spinors. The locus of pure spinors in $\mathbb{P}(S)$ is cut out by a set of quadratic equations in the Cartan coordinates, called Cartan-Pl\"ucker relations (c.f. \cite{Cartan1966, Chevalley1954}). Since $\text{dim }\mathbb{P}(S) \gg \text{dim OG}(V)$ when $n$ is large, these relations can be quite complicated. Our construction of $\tau$-function on the other hand, gives a constructive description of E. Cartan's pure spinors, i.e. a parameterization of the so-called spinor varieties (c.f. \cite{Anderson-Nigro2020, Manivel2009}).
\begin{theorem}\label{thm:purespinor}
The $\tau$-function defined in Theorem \ref{thm:taufunctioninf} has a Schur-Q expansion. For strict partitions $\lambda$, the coefficients of $Q_{\lambda}$ in this expansion satisfy all the Cartan-Pl\"ucker relations, i.e. they are Cartan coordinates of a pure spinor. All pure spinors can be obtained in this way. That is for any pure spinor $s$, all its Cartan coordinates are pfaffians of certain minors of a skew-symmetric matrix determined by $s$. 
\end{theorem}

As we mentioned before, BKP hierarchy enjoys some kind of universality property and a fascinating part of BKP theory is its various reductions to other integrable hierarchies. Here by reduction we mean putting some further constraints on a system which are compatible with that system. For example, Sato's theory of KP hierarchy is based on a pseudo-differential operator $L$, and BKP hierarchy itself can be viewed as a reduction of the KP hierarchy by imposing the following constraint on $L:\partial^{-1}L^*\partial = -L$, where $L^*$ is the formal adjoint of $L$. As a result of this constraint, we only have non-trivial compatible odd time flows $\t_B = (t_1, t_3, t_5, \dots)$ in the BKP hierarchy. As another example, the KdV hierarchy is obtained as $2$-reduction of the KP hierarchy in the sense that we require $L^2$ to be a second order differential operator. More generally, the $l$-reduction of the KP hierarchy is the constraint that $L^l$ is an $l$-th order differential operator, and in these cases the corresponding $\tau$-function usually depends on the time variables $(t_{nl}, n > 0)$ in a trivial way (c.f. \cite{Kodama-Xie2021KP}). In this sense it is natural to study the $l$-reduction of the BKP hierarchy for $l$ a positive odd integer since the BKP hierarchy does not have even time flows from the very beginning, and these reductions were investigated in \cite{Date-Jimbo-Kashiwara-Miwa1982}.

Since both KdV hierarchy and BKP hierarchy depend only on odd times flows, it was speculated since the early days after the discovery of BKP hierarchy that there should be a simple relation between these two integrable systems\footnote{Personal communication with Y. Kodama}. It certainly comes as a ``surprise'' when A. Alexandrov announced his result in \cite{Alexandrov2021} that solutions of KdV satisfy BKP hierarchy up to rescaling of the time parameters by $2$. More precisely, we have
\begin{theorem}[\cite{Alexandrov2021}]\label{thm:KdVinBKP}
For any KdV $\tau$-function,
\begin{equation}\label{eq:KdVinBKP}
\tau(\t_B) = \tau_{\text{KdV}}(\t_B \slash 2)
\end{equation}
is a $\tau$-function of the BKP hierarchy.
\end{theorem}

We reprove Theorem \ref{thm:KdVinBKP} by showing that KdV hierarchy is nothing but $4$-reduction of the BKP hierarchy. Several things need to be explained here. Geometrically KP hierarchy describes the group orbit of $\mathcal{A}_{\infty}$ through the highest weight vector of its level one representation, and algebraically $l$-reduction of KP hierarchy means we restrict ourselves to the sub-Lie algebra $\widehat{\mathfrak{sl}}_l \subset \mathfrak{a}_{\infty}$ and consider the corresponding subgroup orbit. Of course, in general the Kac-Moody algebra $\widehat{\mathfrak{sl}}_l$ is not a sub-Lie algebra of $\mathfrak{b}_{\infty}$, so the proper sub-Lie algebra of $\mathfrak{b}_{\infty}$ we should consider for $l$-reduction is $\widehat{\mathfrak{sl}}_l \cap \mathfrak{b}_{\infty}$. When $l = 4$, we have $A_1^{(1)} \cong D_2^{(2)} = \widehat{\mathfrak{sl}}_4 \cap \mathfrak{b}_{\infty}$, and the level one representation of $\widehat{\mathfrak{sl}}_2 \subset \mathfrak{a}_{\infty}$ goes to the spin representation of $D_2^{(2)} \subset \mathfrak{b}_{\infty}$. The rescaling of time parameters comes exactly from this correspondence. Equipped with this understanding we can give a complete characterization of the KdV orbits inside the BKP hierarchy (see Section \ref{sec:KdVBKP} for more details). A simple consequence of this analysis takes the following form:
\begin{theorem}
For any sequence of complex numbers $\mathbf{a}=(a_1, a_3, a_5, \dots)$, the following formal series is a $\tau$-function for the KdV hierarchy
\[\tau_{\text{KdV}}({t_1}, {t_3}, \dots) = \sum \limits_{\lambda \in \text{DP}}Q_{\lambda}(\mathbf{a})Q_{\lambda}(2\t_B).\]
\end{theorem}

The rest of the paper is organized as follows. In Section \ref{sec:Background} we collect some background information on the f-KT hierarchy and the KP type hierarchies. In Section \ref{sec:BninDn1} we give a concrete presentation of the Lie algebra $\mathfrak{so}_{2n+1}$ which is more suitable for our purpose. In Section \ref{sec:genericsolution} we give a parameterization for the big Schubert cell of the orthogonal Grassmannian manifold based on the presentation of $\mathfrak{so}_{2n+1}$ we introduced in Section \ref{sec:BninDn1} and with this parameterization we prove Theorem \ref{thm:taufunctiongeneric}. In Section \ref{sec:generalcase} we classify and parameterize all the Schubert cells of the orthogonal Grassmannian manifold, and complete the proof of Theorem \ref{thm:normalforminf}, Theorem \ref{thm:taufunctioninf}, Theorem \ref{thm:purespinor} and derive some other consequences of these results. In Section \ref{sec:KdVBKP}, we reprove Theorem \ref{thm:KdVinBKP}, and completely characterize the KdV orbits inside the BKP hierarchy.

\medskip

\noindent
{\bf Acknowledgements}
This paper is a continuation of our endeavor in \cite{Kodama-Xie2021f-KT} to understand structure of solutions of f-KT hierarchy, the author would like to thank Yuji Kodama for his guidance, collaboration and his interests in this work. He also thanks Xiangke Chang for bringing the reference \cite{Ishikawa-Wakayama1995} to his attention and for the useful discussions at the beginning stage of the present work. He appreciated Chang's invitation and the financial support by State Key Laboratory of Scientific and Engineering Computing during his stay in Chinese Academy of Sciences. The author would like to thank Youjin Zhang's, Xiaomeng Xu's and Aleksandr Yu. Orlov's encouragements and interests in this work. The author also would like to thank Yu Li, Yan Zhou, Alexander Alexandrov, David E. Anderson, Kanehisa Takasaki, Xiaobo Liu, John Harnad and Guo Chuan Thiang for some useful feedbacks and comments. This work is partially supported by the National Key Research and Development Program of China (No. 2021YFA1002000) and by the Boya Postdoctoral Fellowship of Peking University.


\section{The full Kostant-Toda lattice and KP type hierarchies}\label{sec:Background}

\subsection{The full Kostant-Toda lattice in general}\label{sec:f-KThierarchy}
The full Kostant-Toda (f-KT) lattice can be abstractly defined on any simple Lie algebra as follows. Let $\mathfrak{g}$ be a complex simple Lie algebra of rank $n$, $\mathfrak{h}$ a Cartan subalgebra of $\mathfrak{g}$. Choose a set of simple roots $\Pi = \{\alpha_1, \alpha_2, \dots, \alpha_{n}\}$, and denote by ${\Sigma}_{\pm}$ the sets of positive and negative roots respectively and $\Sigma = \Sigma_+ \cup \Sigma_-$. Let $\{H_i, X_i, Y_i\}$ be a Chevalley basis of $\mathfrak{g}$, then
\[
[H_i, H_j]=0,\qquad [H_i, X_j]=C_{i,j}X_j,\qquad [H_i, Y_j]=-C_{i,j}Y_j,
\]
where $C_{i,j}$ is the Cartan matrix of $\frak{g}$. We can take a Chevalley system $(X_{\alpha})_{\alpha \in \Sigma}$ of $(\mathfrak{g}, \mathfrak{h})$ such that $X_i=X_{\alpha_i}$ and $Y_i=Y_{\alpha_i}=X_{-\alpha_i}$ and $C_{i,j}=\alpha_i(H_j)$ (see, e.g. \cite{Bourbaki2005}). We sometimes denote $Y_{\alpha} := X_{-\alpha}$ for $\alpha \in \Sigma_+$.

Let $\mathfrak{n}_{\pm} = \sum \limits_{\alpha \in \Sigma_{\pm}}\mathbb{C}X_{\alpha}$ and $\mathfrak{b}_{\pm} = \mathfrak{h} + \mathfrak{n}_{\pm}$ be the maximal nilpotent subalgebras and Borel subalgebras of $\mathfrak{g}$ and $\mathcal{N}_{\pm}, \mathcal{B}_{\pm}$ the corresponding Lie groups, respectively. Then $\mathfrak{g}$ admits the following decomposition,
\begin{align*}
\mathfrak{g} = \mathfrak{n}_- \oplus \mathfrak{h} \oplus \mathfrak{n}_+ = \mathfrak{n}_- \oplus \mathfrak{b}_+.
\end{align*}

Let $e = \sum \limits_{i = 1}^{n} X_i$, then the Lax matrix $L_{\mathfrak{g}}$ is defined as
\begin{equation}\label{eq:Lax}
L_{\mathfrak{g}}= e + \sum_{i=1}^{n} a_i(\mathbf{t})H_i + \sum_{\alpha\in\Sigma_+}b_{\alpha}(\mathbf{t})Y_\alpha,
\end{equation}
where  $a_i(\mathbf{t})$ and $b_{\alpha}(\mathbf{t})$ are functions of the multi-time variables $\mathbf{t}=(t_{m_k}:k=1,2,\ldots, n)$, and $m_k$ are the Weyl exponents (c.f. \cite{Bourbaki2002}). 
For each time variable, we have the f-KT hierarchy defined by
\begin{equation}\label{eq:fKT}
\frac{\partial L_{\mathfrak{g}}}{\partial t_{m_k}}=[B_{m_k}, L_{\mathfrak{g}}],\qquad \text{with}\quad B_k=\Pi_{\mathfrak{b}_+}\nabla I_{k+1},\end{equation}
where  $\nabla$ is the gradient with respect to the Killing form $K$, i.e. for any $x\in \mathfrak{g}$, $dI_k(x)=K(\nabla I_k,x)$,
and $\Pi_{\mathfrak{b}_+}$ represents the projection from $\mathfrak{g}$ to $\mathfrak{b}_+$ with kernel $\mathfrak{n}_-$.
Here the functions $I_k=I_k(L_{\mathfrak{g}})$ are Chevalley invariants which for example in type $A$ are defined by
\[
I_{k+1}(L_{\mathfrak{g}})=\frac{1}{k+1}\text{tr}(L_{\mathfrak{g}}^{k+1}),\quad\text{which gives}\quad \nabla I_{k+1}=L_{\mathfrak{g}}^k.
\]

It is known that the f-KT hierarchy is completely integrable \cite{Ercolani-Flaschka-Singer1993, Gekhtman-Shapiro1999}, and their solutions can be expressed in terms of the so-called $\tau$-functions which are defined as follows (c.f. \cite{Xie2022}). We denote by $(\rho_i, V^{\omega_i}), 1 \le i \le n$, the $i$-th fundamental representation (finite dimensional) of $\mathfrak{g}$ with highest weight vector $v^{\omega_i}$. Let $\langle \cdot, \cdot \rangle$ be a scalar Hermitian product on $V^{\omega_i}$ so that the weight vectors form an orthonormal basis; moreover, we require that the operators $\rho_i(X_{\alpha})$ and $\rho_i(Y_{\alpha})$ are adjoint to each other. It is known that such kind of scalar Hermitian product always exists (c.f. \cite{Kostant1979}).

\begin{definition}
The $i$-th $\tau$-function of the f-KT hierarchy is defined as
\[\tau_i(\mathbf{t}) = \langle v^{\omega_i}, \exp(\Theta_{L_0}(\mathbf{t})) \cdot v^{\omega_i}\rangle,\]
where for $A \in \mathfrak{g}$, we formally define
\[\Theta_A(\mathbf{t}) := \sum \limits_{i=1}^{n}A^{m_i}t_{m_i}.\]
\end{definition}

Note that with our choice of $m_k$, $\exp(\Theta_{L_0}(\mathbf{t}))$ is an element in a Lie group $\mathcal{G}$ with Lie algebra $\mathfrak{g}$, and the action of this group element on $v^{\omega_i}$ is well defined.

\begin{proposition}[\cite{Xie2021, Xie2022}]\label{prop:diagonal}
We have the following formula for the diagonal elements of the Lax matrix in the f-KT hierarchy
\[a_i(\mathbf{t}) = \frac{\partial}{\partial t_1}\ln \tau_i(\mathbf{t}), \qquad \mathbf{t} \ll 1,\]
and all the other coefficients $b_{\alpha}(\mathbf{t})$ are uniquely determined by $a_i(\mathbf{t})$.
\end{proposition}

\begin{remark}
Proposition \ref{prop:diagonal} is proved through the LU-factorization of $\exp(\Theta_{L_0}(\mathbf{t}))$. This factorization can always be performed when $\mathbf{t}$ is small enough as $\exp(\Theta_{L_0}(\mathbf{t}))$ is close to the identity matrix, thus lying in the big Bruhat cell. If some of the $\tau$-functions vanish at a fixed multi-time $\mathbf{t} = \mathbf{t}_*$, then the LU-factorization fails and the f-KT flows enter into a smaller Bruhat cell, i.e. $\exp(\Theta_{L_0}(\mathbf{t})) \in \mathcal{N}_- \dot{w} \mathcal{B}_+$ for some $w \in \mathfrak{W}$ in the Weyl group of $\mathfrak{g}$ such that $w \ne \text{id}$, where $\dot{w} \in \mathcal{G}$ denotes a representative of $w \in \mathfrak{W}$. The set of times $\mathbf{t}_*$ where some of the $\tau_k$ vanish is called the \emph{Painlev\'e divisor}. 
\end{remark}

The local behavior of $\tau$-functions around a Painlev\'e divisor where $\tau_k(\mathbf{t}_*) = 0$ can be analyzed as follows (c.f. \cite{Kodama-Xie2021f-KT}). Setting $\mathbf{t} \to \mathbf{t} + \mathbf{t}_*$, then we have
\[\exp(\Theta_{L_0}(\mathbf{t}_*)) = N_*\dot{w}_*B_* \qquad \text{for some} \quad w_* \in \mathfrak{W},\]
where $N_* \in \mathcal{N}_-, B_* \in \mathcal{B}_+$ and ${w} \in \mathfrak{W}$. It is shown in \cite{Xie2022} that for any point $g\cdot\mathcal{B}_+ \in \mathcal{G} \slash \mathcal{B}_+$ in the flag variety, there exists an $L_0$ such that $N_*\dot{w}_*B_*$ is a representative of $g \cdot \mathcal{B}$ in $\mathcal{G}$.

In summary, a general $\tau$-function has the following form:
\begin{equation}\label{eq:singulartau}
\tau_i(\mathbf{t}) = \langle v^{\omega_i}, \exp(\Theta_{L_0}(\mathbf{t})) N_*\dot{w}_* B_* v^{\omega_i}\rangle, \qquad 1 \le i \le n.
\end{equation}

Now we try to simplify expression \eqref{eq:singulartau} for the $\tau$-functions. First we note the following theorem of Kostant:
\begin{proposition}[\cite{Kostant1978}]\label{thm:Kostantsection}
There exists an $n$-dimensional linear subspace $\mathfrak{s} \subset \mathfrak{b}_-$ such that elements in the affine subspace $e + \mathfrak{s}$ are regular. The map
\[\begin{array}{rcl}
 \mathcal{N}_- \times (e + \mathfrak{s}) & \to &  e + \mathfrak{b}_-\\
 (N, x) & \mapsto & \text{Ad}_{N}x
\end{array}\]
is an isomorphism of affine varieties.
\end{proposition}

\begin{remark}
When $\mathfrak{g} = \mathfrak{sl}_{n+1}(\mathbb{C})$, the choice of $\mathfrak{s}$ may be made so that $e + \mathfrak{s}$ is the affine space of traceless companion matrices. 
\end{remark}

Let $\mathcal{F}_{\Omega}$ be the isospectral variety consisting of Lax matrices with fixed Chevalley invariants $\Omega = \{I_1, \dots, I_{n}\}$, i.e.
\[\mathcal{F}_{\Omega} := \{L \in e + \mathfrak{b}_- \ |\ L \text{ has Chevalley invariants $\Omega$}\}.\]

With any fixed choice of $\mathfrak{s}$, Proposition \ref{thm:Kostantsection} says that for any $L \in \mathcal{F}_{\Omega}$ there exists a unique $N \in \mathcal{N}_-$ and $C_{\Omega} \in e + \mathfrak{s}$ such that 
\[L = N^{-1} C_{\Omega}N,\]
where $C_{\Omega}$ records the spectral data of $L$ only.

Applying Proposition \ref{thm:Kostantsection} to $L_0$, that is assuming $L_0 = N_0^{-1}C_{\Omega}N_0$ and noting that
\[b_* \cdot v^{\omega_i} = d_iv^{\omega_i},\]
where $d_i$ is a constant, we then have the following expression for $\tau$-functions
\begin{align*}
\tau_i(\mathbf{t}) & = \langle v^{\omega_i}, \exp(\Theta_{L_0}(\mathbf{t})) N_*\dot{w}_* B_* v^{\omega_i}\rangle \\
& = \langle v^{\omega_i}, n_0^{-1}\exp(\Theta_{C_{\Omega}}(\mathbf{t})) N_0N_*\dot{w}_* B_* v^{\omega_i}\rangle \\
& = d_i \langle v^{\omega_i}, \exp(\Theta_{C_{\Omega}}(\mathbf{t})) \tilde{N}\dot{w}_* v^{\omega_i}\rangle \qquad 1 \le i \le n,
\end{align*}
where $\tilde{N} = N_0N_* \in \mathcal{N}_-$. One can ignore the overall constants $d_i$'s since the solution of the f-KT hierarchy does not depend on $d_i$. We also note that when all the spectral data of $L_0$ are trivial, the general rational solutions of the f-KT hierarchy are given by the following polynomial $\tau$-functions
\begin{equation}\label{eq:polytau}
\tau_i({\t; g}) = \langle v^{\omega_i}, \exp(\Theta_{e}(\t))g v^{\omega_i}\rangle \qquad g \in \mathcal{G}, \quad 1 \le i \le n.
\end{equation}

\subsection{The full Kostant-Toda lattice in type $A$ and type $B$}\label{sec:f-KTAB}
Now we specialize the above general theory to type $A$ and type $B$ Lie algebras. Let $\ell = n + 1$ in type $A$ and $\ell = 2n + 1$ in type $B$, respectively. With the Chevalley system chosen as in \cite{Kodama-Xie2021f-KT}, 
the Lax matrices $L_A \in \mathfrak{sl}_{n + 1}$  and $L_B \in \mathfrak{so}_{2n + 1}$ have the following Hessenberg form 
\begin{align}
& L =
\begin{pmatrix}
a_{1, 1} & 1 & 0 & \cdots & 0 & 0  \\
a_{2,1}& a_{2, 2} &1& \cdots & 0 & 0\\
a_{3, 1} & a_{3, 2} & a_{3, 3} & \ddots & \vdots & \vdots\\
\vdots &\vdots &\ddots &\ddots & 1 & \vdots \\
a_{\ell-1,1}& a_{\ell-1,2} & \cdots &\cdots&a_{\ell-1, \ell-1} &1\\
a_{\ell,1} & a_{\ell,2}&\cdots&\cdots & a_{\ell, \ell-1} & a_{\ell, \ell}
\end{pmatrix},\nonumber
\end{align}
where
\begin{align}
& \begin{array}{l}
a_{i, i} = a_{i} - a_{i-1} \text{ with } a_0 = a_{n + 1} = 0 
\end{array} \text{ if } L = L_A, \\
& \left\{\begin{array}{l}
a_{i, i} = a_i - a_{i - 1} \quad (1 \le i \le n -1) \text{ with } a_0 = 0, \quad a_{n, n} = 2a_{n} - a_{n - 1}\\
a_{i + k, k} = (-1)^{i+1}a_{2n+2-k, 2n+2-i-k},
\end{array}\right. \text{ if } L = L_B.
\end{align}

For type $A$ and type $B$ Lie algebras of rank $n$, the Weyl exponents $m_k$ are $1, 2, \dots, n$ and $1, 3, 5, \dots, 2n -1$, respectively. So the f-KT hierarchy in type $A$ is defined as
\[\frac{\partial L_A}{\partial t_k} = [B_k, L_A], \qquad B_{k} := \Pi_{\mathfrak{b}^+}\nabla (I_{k+1}) = (L_A^{k})_{\ge 0}. \]
Note that the invariants $I_{k+1} = \frac{1}{k + 1}\text{tr}(L_B^{k+1})$ vanish when $k$ is even. Then the $B$-type f-KT hierarchy is defined only for the odd times $\mathbf{t}_B = (t_1, t_3, t_5, \dots, t_{2n-1})$, and each invariant $I_{2k+2}$ defines the Lax equation,
\[\frac{\partial L_B}{\partial t_{2k + 1}} = [B_{2k + 1}, L_B], \qquad B_{2k+1} := \Pi_{\mathfrak{b}^+}\nabla (I_{2k+2}) = (L_B^{2k+1})_{\ge 0}.\]

Let $V := \mathbb{C}^{\ell}$. 
To study the fundamental representations of $\mathfrak{sl}_{\ell}$ we consider the full wedge space
\[F = \bigoplus \limits_{i = 0}^{n + 1} \bigwedge^i V,\]
formed by the linear combination of exterior product of elements in $V$. Let $E_{i,j}$ denote the matrix with $1$ at the $(i, j)$ entry and $0$ at all the other entries, then these elements form a basis for $\mathfrak{gl}_{\ell}$. The natural action of $\mathfrak{gl}_{\ell }$ on $V$ induces an action on $F$ as follows
\[r(E_{i, j})(e_{i_1} \wedge e_{i_2} \wedge \dots \wedge e_{i_k}) = (E_{i, j}e_{i_1}) \wedge e_{i_2} \wedge \dots \wedge e_{i_k} + e_{i_1} \wedge (E_{i, j}e_{i_2}) \wedge \dots + \dots + e_{i_1} \wedge e_{i_2} \wedge \dots \wedge (E_{i, j}e_{i_k}).\]
$F$ then decomposes into direct sum of irreducible modules $F^{(k)} = \bigwedge^kV, 0 \le k \le \ell$ under this action of $\mathfrak{gl}_{\ell}$. Each $F^{(k)}$ is still irreducible when considered as an $\mathfrak{sl}_{\ell}$ module, and serves as the $k$-th fundamental module of $\mathfrak{sl}_{\ell}$ for $1 \le k \le n$.

 Accordingly, the $\tau$-functions are given by
\[\tau^A_i(\t_A) := \langle e_1 \wedge e_2 \wedge \dots \wedge e_i, e^{\Theta_{L_A^0}(\t_A)}e_1 \wedge e_2 \wedge \cdots \wedge e_i \rangle, \qquad 1 \le i \le n,\]
where $(e_i)_{1 \le i \le \ell}$ is the standard basis of $V$, i.e. $e_i$ is a column vector with $1$ at the $i$-th row and $0$ elsewhere, and $\langle \cdot, \cdot \rangle$ is the standard inner product on the wedge space $\bigwedge^i V, 1 \le i \le n $.

Let $I = \{i_1 < i_2 < \dots < i_k\}, J = \{j_1 < j_2 < \dots < j_k\} \subset [\ell] = \{1, 2, \dots, \ell\}$ be two index sets, then we define $\Delta_{I, J}(g) := \langle e_{i_1} \wedge \dots \wedge e_{i_k}, g \cdot e_{j_1} \wedge \dots \wedge e_{j_k}\rangle$. We have

\begin{theorem}[\cite{Kodama-Xie2021f-KT}, \cite{Xie2021}, \cite{Kodama-Williams2015}]
For $g \in \mathcal{SL}_{\ell}$, the $k$-th polynomial $\tau$-function in type $A$ has the following form
\begin{align}\label{eq:SchurexpansionA}
\tau^A_k(\t_A; g) & = \langle e_1 \wedge e_2 \wedge \dots \wedge e_k, e^{\Theta_{e}(\t_A)}g \cdot e_1 \wedge e_2 \wedge \cdots \wedge e_k \rangle \nonumber\\
& = \sum \limits_{I \in \binom{[\ell]}{k}}\Delta_{I, [k]}(g)S_I(\t)
\end{align}
where $S_{I}(\t)$ is the Schur function associated with the index set $I = \{i_1 < i_2 < \cdots < i_k\} \subset [\ell] = \{1, 2, \dots, \ell\}$.
\end{theorem}

\begin{remark}\label{rem:grouporbit}
Here we make a remark which is valid for f-KT hierarchy on all simple Lie algebras but is more easily seen in type $A$. Note that the set of coefficients $(\Delta_{I, [k]}(g))$ in the expansion \eqref{eq:SchurexpansionA} is a set of coordinates for the element $[g \cdot v^{\omega_k}]$ in the projective space $\mathbb{P}(V^{\omega_k})$. From equation \eqref{eq:SchurexpansionA} we see that each f-KT flow is tangent to the group orbit of $\mathcal{G} \cdot v^{\omega_k}$ in $\mathbb{P}(V^{\omega_k})$, and it deforms the above coordinates of $[g \cdot v^{\omega_k}]$ thus serves as a dynamics on the corresponding Grassmannian manifold. Putting all $\tau_k(\t)$'s together, f-KT flows give a dynamical system on the flag variety. It actually can be shown that the f-KT hierarchy completely characterizes these group orbits. In the type $A$ case, this means the following: $f_k(\t) = \sum_{I \in \binom{[n]}{k}}c_{I}S_{I}(\t), 1 \le k \le n$ are $\tau$-functions of the f-KT hierarchy if and only if there exists a $g \in \mathcal{SL}_{\ell}$ such that $c_{k, \lambda} = \Delta_{I, [k]}(g)$, equivalently if and only if $c_{I}$ satisfy all the relevant Pl\"ucker relations (c.f. \cite{Marsh-Rietsch2004}). 
\end{remark}

For type $B$ Lie algebras, from now on we shift the index of $L_B$ by $n$ to accommodate the usual convention in the BKP hierarchy. That is we number the rows and columns in $L_B$ from top to bottom and left to right, respectively, by $(-n, -n + 1, \dots, -1, 0, 1, \dots, n -1, n)$. With this convention, the last $n - 1$ fundamental representations of $\mathfrak{so}_{2n + 1}$ still come from the action of $\mathfrak{so}_{2n + 1}$ on $\bigwedge^i V$ induced from the natural action of $\mathfrak{so}_{2n + 1}$ on $V$, but the first one comes from the so-called spin representation $(\rho, S)$ (see Section \ref{sec:Spin}).
We formally define 
\[\Theta_{L^0_B}(\mathbf{t}_B) := \sum \limits_{k = 1}^{n}(L^0_B)^{2k-1}t_{2k-1}, \qquad \t_B = (t_1, t_3, \dots, t_{2n-1}).\]
Then the $\tau$-functions are given by
\begin{align*}
& \tau^B_i(\mathbf{t}_B) := \langle e_{-n} \wedge e_{-n + 1} \wedge \cdots \wedge e_{-i}, e^{\Theta_{L^0_B}(\mathbf{t}_B)}e_{-n} \wedge e_{-n + 1} \wedge \cdots \wedge e_{-i}\rangle, \qquad 2 \le i \le n;\\
& \tau^B_{1}(\mathbf{t}_B) := \langle v^{\omega_1}, \rho({e^{\Theta_{L^0_B}(\mathbf{t}_B)}}) \cdot v^{\omega_1} \rangle,
\end{align*}
where $v^{\omega_1} \in S$ is the highest weight vector. 

For most part of the present paper, we will focus on the polynomial $\tau^B_1(\t_B; g)$ and we will denote it by $\tau_{\text{BKT}}(\t_B; g)$ henceforth, that is
\begin{equation}\label{eq:fKTtauB}
\tau_{\text{BKT}}(\t_B; g) = \langle v^{\omega_1}, \rho(\exp(\Theta_{e}(\t_B)g)) \cdot v^{\omega_1}\rangle, \qquad g \in \mathcal{SO}_{2n+1}.
\end{equation}

We have the following alternative expression for $\tau^2_{\text{BKT}}(\t_B; g)$ relating it with $\tau$-function for type $A$ f-KT hierarchy.
\begin{theorem}[\cite{Kodama-Xie2021f-KT}]
For any $g \in \mathcal{SO}_{2\ell + 1}(\mathbb{C})$, we have
\begin{align*}
\tau_{\text{BKT}}^2(\t_B; g) & = \langle e_{-n} \wedge e_{-n + 1} \wedge \dots \wedge e_{-1}, \exp(\Theta_e(\t_B)) {g} e_{-n} \wedge e_{-n + 1} \wedge \dots \wedge e_{-1}\rangle\\
& = \langle e_{-n} \wedge e_{-n + 1} \wedge \dots \wedge e_{-1} \wedge e_0, \exp(\Theta_e(\t_B)) {g} e_{-n} \wedge e_{-n + 1} \wedge \dots \wedge e_{-1} \wedge e_0\rangle.
\end{align*}
\end{theorem}


\subsection{KP and BKP hierarchies and their $\tau$-functions}\label{sec:ABKP}
In Section \ref{sec:f-KTAB} we saw that geometrically f-KT hierarchy describes the group orbits $\mathcal{G} \cdot v^{\omega_i}, 1 \le i \le n$. Similar things can be considered for infinite-dimensional matrix algebras and the corresponding homogeneous varieties. The standard references for this section are \cite{Kac1990, Kac-Raina1987, Miwa-Jimbo-Date2000, Jimbo-Miwa1983}.

Let
\begin{align*}
& \mathfrak{a}_{\infty} = \bar{\mathfrak{a}}_{\infty} \oplus \mathbb{C}c\\
& \mathfrak{b}_{\infty} = \bar{\mathfrak{b}}_{\infty} \oplus \mathbb{C}c,
\end{align*}
which are central extensions of
\begin{align*}
& \bar{\mathfrak{a}}_{\infty} = \{(a_{i, j}) \ |\ a_{i, j} = 0 \text{ for } |i - j| \gg 0\}\\
& \bar{\mathfrak{b}}_{\infty} = \{(a_{i,j}) \in \bar{\mathfrak{a}}_{\infty}\ |\ a_{i, j} = (-1)^{i+j+1}a_{-j, -i}\}.
\end{align*}

We still let $E_{i, j}$ denote the matrix with $1$ at the $(i, j)$ entry and $0$ at all the other entries. Then a typical element in $\bar{\mathfrak{a}}_{\infty}$ is a finite linear combination of matrices of the form 
\[a_k = \sum \limits_{i \in \mathbb{Z}}\lambda_i E_{i, i+k}\]
where the $\lambda_i$'s are arbitrary complex numbers. In particular, we have
\[H_k = \sum \limits_{i \in \mathbb{Z}}E_{i, i+k} \in \bar{\mathfrak{a}}_{\infty}.\]
Elements in $\bar{\mathfrak{a}}_{\infty}$ acts in a usual way on the infinite dimensional complex vector space $V := \mathbb{C}^{\infty}$ of columns $(c_i)_{i \in \mathbb{Z}}$, where all but a finite number of the $c_i$'s are zero, that is
\[V = \bigoplus \limits_{j \in \mathbb{Z}} \mathbb{C}e_j,\]
and
\[a e_j = \sum \limits_{i \in \mathbb{Z}}a_{i,j}e_i, \qquad a \in \bar{\mathfrak{a}}_{\infty}.\]

Consider the full wedge space defined by
\[F = \bigwedge^{\infty} V,\]
then the Lie algebra $\mathfrak{a}_{\infty}$ acts on $F$ by the following rule:
\begin{align*}
& \bar{r}(E_{i, j})(e_{i_1} \wedge e_{i_2} \wedge \dots) = -\theta(i \le 0)\delta_{i, j}e_{i_1} \wedge e_{i_2} \wedge \dots + (E_{i, j}e_{i_1}) \wedge e_{i_2} \wedge \dots + e_{i_1} \wedge (E_{i, j}e_{i_2}) \wedge \dots + \dots,\\
& \bar{r}(c)(e_{i_1} \wedge e_{i_2} \wedge \dots) = e_{i_1} \wedge e_{i_2} \wedge \dots.
\end{align*}
where
\[\theta(i \le 0) = \left\{\begin{array}{ll}
1 \qquad & \text{if } i \le 0,\\
0 & \text{otherwise.}
\end{array}\right.\]

We then have
\begin{align*}
& [\bar{r}(E_{i, j}), \bar{r}(E_{k, l})] = 0 \qquad \text{for } j \ne k, i \ne l\\
& [\bar{r}(E_{i, j}), \bar{r}(E_{j, l})] = \bar{r}(E_{i, l}) \qquad \text{for } i \ne l\\
& [\bar{r}(E_{i, j}), \bar{r}E_{k, i}] = -\bar{r}(E_{k, j}) \qquad \text{for } j \ne k\\
& [\bar{r}(E_{i, j}), \bar{r}(E_{j, i})] = \bar{r}(E_{i, i}) - \bar{r}(E_{j, j}) + \alpha(E_{i, j}, E_{j, i})I.
\end{align*}
where
\begin{align}\label{eq:twocycle}
& \alpha (E_{i, j}, E_{j, i}) = -\alpha(E_{j, i}, E_{i, j}) = 1 \qquad \text{if } i \le 0, \  j \ge 1,\\
& \alpha(E_{i, j}, E_{k, l}) = 0 \qquad \text{in all other cases.}\nonumber
\end{align}

$F$ decomposes into direct sum of irreducible modules under this action of $\mathfrak{a}_{\infty}$
\[ F = \bigoplus \limits_{m \in \mathbb{Z}}F^{(m)}, \]
where $F^{(m)}$ is the vector space with basis consisting of expressions of the form
\[v = e_{i_m} \wedge e_{i_{m-1}} \wedge \cdots,\]
such that $i_m > i_{m-1} > \cdots$ and $i_k = k + m$ for $k \ll 0$. Note that $F^{(m)}$ is a highest wight module with highest weight vector
\[v_m = e_m \wedge e_{m-1} \wedge e_{m-2} \wedge \cdots \in F^{(m)}.\]
Denote the restriction of $\bar{r}$ on $F^{(m)}$ by $\bar{r}_m$, and let $\nu_s$ be the shift operator defined by $\nu_s(e_j) = e_{j - s}, j \in \mathbb{Z}$, then we have
\[\nu_s \bar{r}_s \nu_s^{-1} = \bar{r}_0.\]
Thus for $\mathfrak{a}_{\infty}$ all the $\bar{r}_m$'s are equivalent. In particular, the representation $(r_0, F^{(0)})$ is usually called the basic representation.

Consider the groups defined as
\begin{align*}
& \mathcal{GL}_{\infty} = \{A = (a_{i, j})_{i, j \in \mathbb{Z}}\ |\ \text{$A$ is invertible, $a_{ij} = \delta_{ij}$ for almost all $i, j$}\},\\
& \mathcal{O}_{\infty} = \{A \in \mathcal{GL}_{\infty} \ |\ \text{$A$ preserves the form defined by $\langle e_i, e_j \rangle = (-1)^i\delta_{i+j, 0}$}\}.
\end{align*}
$\mathcal{GL}_{\infty}$ has a natural action $R_0$ on $F^{(0)}$ defined as follows
\begin{equation*}
R_0(A)(e_{i_1} \wedge e_{i_2} \wedge \dots) = Ae_{i_1} \wedge Ae_{i_2} \wedge \dots.
\end{equation*}

Let
\[H_A(\t) := \sum \limits_{k > 0}t_kH_k,\]
and note that the action $r_0$ of $H_A(\t)$ on $V$ can be lifted to the action $R_0$ of $\exp(H_A(\t))$ on certain completion $\bar{V}$ of $V$, then as in the finite-dimensional case we can similarly define $\tau$-function as follows (which is the same as formula \eqref{eq:tauKPintro} restricted to full rank $\xi$ with finite nontrivial elements)
\[\tau_{\text{KP}}(\t; g) = \langle v_0, R_0(\exp(H_A(\t))g)\cdot v_0\rangle, \qquad g \in \mathcal{GL}_{\infty}.\]
Under the so-called the Boson-Fermion correspondence 
\begin{equation}\label{eq:Bonson-FermionA}
\begin{array}{rcl}
\sigma^A_0: F^{(0)} & \longrightarrow & \mathbb{C}[t_1, t_2, \dots]\\
 v_0 & \mapsto & 1\\
 H_n & \mapsto & \frac{\partial}{\partial t_n}, \qquad n \ge 1\\
 H_{-n} & \mapsto & nt_n, \qquad n \ge 1,
\end{array}
\end{equation}
points in the group orbit $\mathcal{O} = \mathcal{GL}_{\infty} \cdot v_0$ has the following characterization
\begin{theorem}[\cite{Kac1990}, \cite{Miwa-Jimbo-Date2000}]
A nonzero element $\tau(\t)$ of $\mathbb{C}[t_1, t_2, \dots]$ is contained in $\sigma^A_0(\mathcal{O})$ if and only if
\begin{equation}\label{eq:bilinearKP}
\oint\left(\exp\sum \limits_{j \ge 1}2z^j s_j\right)\left(\exp - \sum \limits_{j \ge 1}\frac{z^{-j}}{j}\frac{\partial}{\partial s_j}\right)\tau(\t+\mathbf{s})\tau(\t-\mathbf{s}) \frac{dz}{2\pi i} = 0,
\end{equation}
where $\mathbf{s} = (s_1, s_2, \dots)$ and the integration is taken along a small circle around $z = 0$.
\end{theorem}
The system \eqref{eq:bilinearKP} is the Hirota bilinear form of KP hierarchy. Let $u(\t) = 2\partial_1^2\ln \tau(\t)$, then $u(\t)$ satisfies the classical Kadomtsev-Petviashvili (KP) equation
\begin{equation}\label{eq:KP}
3\partial_2^2u + \partial_1(-4\partial_3 u+6u\partial_1u+\partial_1^3u)=0,
\end{equation}
with the partial derivatives $\partial_k^iu := \frac{\partial^iu}{\partial t_k^i}$. 

\begin{remark}
To obtain more general solutions of the KP hierarchy, we can consider a larger infinite matrix group $\bar{\mathcal{A}}_{\infty} \supset \mathcal{GL}_{\infty}$ which projectively acts on a certain completion $\bar{F}$ of $F$. $\bar{F}$ decomposes into unions of infinite Grassmannian manifolds under the linear action of the central extension $\mathcal{A}_{\infty}$ of $\bar{\mathcal{A}}_{\infty}$. Since $\tau_{\text{KP}}(\t)$ as defined in \eqref{eq:tauKPintro} is a section of the determinant bundle on $\bar{F}$, the group $\mathcal{A}_{\infty}$ is required to be small enough so that this section makes sense before and after the action of $\mathcal{A}_{\infty}$, and to be large enough so that its Lie algebra contains $\mathfrak{a}_{\infty}$. See \cite{Sato-Sato1982}, \cite{Segal-Wilson1985} and \cite{Arbarello-Concini-Kac-Procesi1988} for several different constructions of $\bar{F}$ and $\mathcal{A}_{\infty}$. Assuming the existence of this group $\mathcal{A}_{\infty}$, we denote its representation on $\bar{F}$ by $\bar{R}$ in the following. Since all Lie algebras we discuss in this paper are sub-Lie algebras of $\mathfrak{a}_{\infty}$, the existence of their Lie groups and completion of the corresponding modules are also assumed.
\end{remark}

The $\tau$-function for BKP hierarchy is, in a similar way, associated with the infinite dimensional Lie algebra $\mathfrak{b}_{\infty}$ and its fundamental representation.
More precisely, let $\t_B = (t_1, t_3, t_5, \dots)$, and note that $H_k \in \bar{\mathfrak{b}}_{\infty}$ for $k$ odd. Let
\[H_B(\t_B) = \sum \limits_{k >0, \text{ odd}}t_kH_{k},\]
then $\tau$-function for the BKP hierarchy is given by
\[\tau_{\text{BKP}}(\t_B; g) = \langle v^{\omega_1}, \rho(\exp H_B(\t_B) g) \cdot v^{\omega_1}\rangle, \qquad g \in \mathcal{O}_{\infty},\]
where $(\rho, S)$ is the infinite-dimensional version of the spin representation with highest weight vector $v^{\omega_1}$ (c.f. \cite{Date-Jimbo-Kashiwara-Miwa1982}). 

The Boson-Fermion correspondence in type $B$ takes the following form
\begin{equation}\label{eq:Bonson-FermionB}
\begin{array}{rcl}
\sigma^B_0: S & \longrightarrow & \mathbb{C}[t_1, t_3, \dots]\\
 v^{\omega_1} & \mapsto & 1\\
 H_n & \mapsto & \frac{\partial}{\partial t_n}, \qquad n \ge 1\\
 H_{-n} & \mapsto & \frac{1}{2}nt_n, \qquad n \ge 1.
\end{array}
\end{equation}

We have the following important relationship between $\tau$-functions of the KP hierarchy and BKP hierarchy.
\begin{theorem}[\cite{Date-Jimbo-Kashiwara-Miwa1982, You1989}]\label{thm:ABKPrelation}
The $\tau$-functions for BKP and KP hierarchies are related as follows: for any $g \in \mathcal{O}_{\infty}$, we have
\[\tau_{\text{KP}}(t_1, 0, t_3, 0, \dots; {g}) = \tau^2_{\text{BKP}}(t_1, t_3, \dots; g).\]
\end{theorem}

As before, to obtain more general solutions for the BKP hierarchy we can also consider a larger infinite matrix group $\mathcal{B}_{\infty}$ of type $B$, such that for any $g \in \mathcal{B}_{\infty}$ there is a well-defined action on a certain completion $\bar{S}$ of $S$ and the corresponding $\tau_{\text{BKP}}(\t_B; g)$ makes sense. Such $\mathcal{B}_{\infty}$ is also required to have a well-defined Lie algebra containing $\mathfrak{b}_{\infty}$ (c.f. \cite{Dirac1974, Plymen-Robinson1994, Pressley-Segal1986} for several different analytic constructions of the Fock spaces and the spin representations of $\mathfrak{b}_{\infty}$).

Finally, from the constructions in the last two sections, we have
\begin{corollary}
All the $\tau$-functions for f-KT hierarchy in type $A$ solves KP hierarchy, and all the $\tau$-functions $\tau^B_1(\t_B)$ for f-KT hierarchy in type $B$ are solutions to the BKP hierarchy.
\end{corollary}


\section{$B_{n}$ as a sub-Lie algebra of $D_{n+1}$}\label{sec:BninDn1}
The main goal in the following a few sections is to find a simple expression for $\tau_{\text{BKT}}(\t_B)$, the $\tau$-function of the f-KT hierarchy in type $B$ associated with the spin representation, and for that purpose we use a slightly different realization of the Lie algebra $B_{n}$ than the standard references (c.f. \cite{Bourbaki2005, Fulton-Harris2013}). The basic idea is that the Lie algebra of type $B_{n}$ could be realized as a sub-Lie algebra of $D_{n}$ in $(2n+2) \times (2n+2)$ matrices. We should point out that this is not something new at all in the study of simple Lie algebras, for example in \cite{Cartan1966} Cartan first studied type $B$ Lie algebra $\mathfrak{so}_{2n+1}$, and the theory of type $D$ Lie algebra $\mathfrak{so}_{2n}$ is obtained by restricting the associated quadratic form on $\mathbb{C}^{2n + 1}$ to $\mathbb{C}^{2n}$; on the other hand, in \cite{Chevalley1954}, Chevalley first studied type $D$ Lie algebra $\mathfrak{so}_{2n+2}$, and the theory of type $B$ Lie algebra $\mathfrak{so}_{2n+1}$ is deduced from that of type $D$. Moreover, it is well-known that type $B$ Lie algebras could be constructed from the simply laced type $D$ Lie algebras by a folding procedure (c.f. \cite{Kac1990}, \cite{Slodowy1980}).


\subsection{$D_{n+1}$ as $(2n + 2) \times (2n + 2)$ matrices}\label{sec:in2n+2}
The setup is similar to the one presented in \cite{Kodama-Xie2021f-KT, Kodama-Okada2023}. Let $\tilde{V}$ be a $(2n + 2)$-dimensional vector space, and $\tilde{\Psi}$ a non-degenerate symmetric bilinear form of maximal index $n+1$ on $\tilde{V}$. Denote by $\tilde{Q}$ the quadratic form associated with $\tilde\Psi$, i.e., $\tilde{Q}(x) = \frac{1}{2}\tilde\Psi(x, x)$ for $x \in \tilde{V}$. Then $\tilde{V} = \tilde{F} \oplus \tilde{F}'$, where $\tilde{F}$ and $\tilde{F}'$ are two maximal totally isotropic subspaces with respect to $\tilde{Q}$. Note that $\tilde{F}$ and $\tilde{F}'$ are in duality via $\tilde\Psi$. Let $(e_{-i})_{0 \le i \le n}$ be a basis of $\tilde{F}$ and $(e_{+i})_{0 \le i \le n}$ the dual basis of $\tilde{F}'$ such that $\tilde \Psi(e_i, e_{-j}) = (-1)^i\delta_{i, j}$ \footnote{Note that here we view $e_{+0}$ and $e_{-0}$ as different vectors, and the $\Psi$ we take is slightly different from \cite{Bourbaki2005}.}. Then
\begin{equation}\label{eq:BasisD}
\tilde{\Xi} = (e_{-n}, e_{-n+1}, \dots, e_{-1}, e_{-0}, e_{+0}, e_{+1}, \dots, e_{+n})
\end{equation}
is a basis of $\tilde{V}$; we have
\[\tilde{Q}(\sum x_{\alpha}e_{\alpha}) =\sum \limits_{i = 0}^n (-1)^ix_ix_{-i}\]
and $\tilde\Psi$ with respect to this basis takes the form of the following $(2n+2) \times (2n+2)$ square matrix (i.e. $\tilde{M}$ is the Gram matrix of $\tilde{\Psi}$)
\begin{equation}\label{eq:tildeM}
\tilde{M}=\begin{pmatrix}
0 &\cdots & 0 & 0  & 0 & 0 &  \cdots & (-1)^n   \\
\vdots &\udots & &  &\vdots & & \udots&\vdots\\
0& & 0 & 0 & 0 &-1 &  &0\\
0 & & 0 & 0 & 1 & 0 & & 0\\
0&\cdots & 0& 1 & 0 & 0 & \cdots & 0\\
0&  &-1&0&0&0&&0\\
\vdots &\udots& &\vdots & & &\udots &\vdots\\
(-1)^n &\cdots&0 &0&0& 0 &\cdots &0
\end{pmatrix} = \begin{pmatrix}[ccc|c|ccc]
 &    & &    &  & & \\
 & 0 & & 0 &  & J &\\
 &  &  &   & & & \\
\hline 
& \raisebox{-0.05cm}{$ 0 $} & & \sigma & &\raisebox{-0.05cm}{$ 0 $} & \\
\hline
 &    & &    &  & & \\
 & J^T & & 0 &  & 0  &\\
 &  &  &   & & & \\
\end{pmatrix},
\end{equation}
where 
\[\sigma = \begin{pmatrix}
 & 1 \\ 1 & 
\end{pmatrix} \qquad \text{and} \qquad J = \begin{pmatrix}
& & & (-1)^n\\
& & \udots & \\
& 1 & & \\
-1 & & & 
\end{pmatrix}.\]
We denote by $\mathfrak{so}_{2n+2}(\mathbb{C})$ the orthogonal Lie algebra associated with $\tilde\Psi$ and Cl$(\tilde{V}, \tilde{Q})$ the Clifford algebra of $\tilde{V}$ relative to $\tilde{Q}$.

We take the following Chevalley generators for $\mathfrak{so}_{2n+2}(\mathbb{C})$:
\[\begin{array}{rcl}
{H}_{\tilde\alpha_i} & = & \tilde{E}_{-i, -i} - \tilde{E}_{-(i-1), -(i-1)} + \tilde{E}_{+(i-1), +(i-1)} - \tilde{E}_{+i, +i}, \qquad 1 \le i \le n,\\
{H}_{\tilde\alpha_0} & = & \tilde{E}_{-1, -1} + \tilde{E}_{-0,-0} - \tilde{E}_{+0, +0} - \tilde{E}_{+1, +1},
\end{array}\]
where $\tilde\alpha_n = \varepsilon_n - \varepsilon_{n-1}, \dots, \tilde\alpha_2 = \varepsilon_2 - \varepsilon_1, \tilde\alpha_1 = \varepsilon_1 - \varepsilon_{0}, \tilde\alpha_0 = \varepsilon_1 + \varepsilon_{0}$ and $\varepsilon_i \in \mathfrak{h}_D^*$ is defined by $\varepsilon_i(\tilde{E}_{-j, -j} - \tilde{E}_{-k, -k}) = \delta_{i,j} - \delta_{i,k}$ for $0 \le i, j, k \le n$.
We take the simple root vectors,
\[\begin{array}{rcl}
{X}_{\tilde\alpha_i} & = & \tilde{E}_{-i, -(i-1)} + \tilde{E}_{+(i -1), +i} \qquad 1 \le i \le n,\\
{X}_{\tilde\alpha_0} & = & \tilde{E}_{-1, +0} + \tilde{E}_{-0, +1}
\end{array}\]
and the negative root generators,
\[\begin{array}{rcl}
{Y}_{\tilde\alpha_i} & = & {X}_{-\tilde\alpha_i} = \tilde{E}_{-(i-1), -i} + \tilde{E}_{+i, +(i-1)} \qquad 1 \le i \le n,\\
{Y}_{\tilde\alpha_0} & = & {X}_{-\tilde\alpha_0} = \tilde{E}_{+0, -1} + \tilde{E}_{+1, -0}.
\end{array}\]
The other negative root vectors are then generated by taking the commutators of simple roots
\[[{Y}_{\tilde\alpha}, {Y}_{\tilde\beta}] = {N}_{\tilde\alpha, \tilde\beta}{Y}_{\tilde\alpha + \tilde\beta} \qquad \text{if} \quad \tilde\alpha, \tilde\beta, \tilde\alpha + \tilde\beta \in \tilde\Sigma_+,\]
where the set of positive roots is given by

\[\tilde{\Sigma}_+ = \left\{\begin{array}{cl}
\tilde\alpha_0 + \tilde\alpha_2 + \cdots + \tilde\alpha_i & (1 \le i \le n),\\
\tilde{\alpha}_0 + \tilde\alpha_1 + 2\tilde\alpha_2 + \cdots + 2\tilde\alpha_{i} + \tilde\alpha_{i+1} + \cdots + \tilde\alpha_j  \qquad & (1 \le i < j \le n),\\
\tilde\alpha_{i+1} + \cdots + \tilde\alpha_{j} & (0 \le i < j \le n). 
\end{array}
\right\}\]
We have $|\tilde\Sigma_+| = n(n+1)$.

The Weyl group $\mathfrak{W}_{D_{n+1}}$ associated with the algebra of type $D_{n+1}$ can be given as follows
\[s_i^D = \left\{\begin{array}{ll}
{s}_{-1}{s}_{-0}{s}_{-1}{s}_{-0}{s}_{+0}{s}_{-0}, \qquad & i = 0,\\
{s}_{-i}{s}_{+(i-1)},  & 1 \le i \le n,
\end{array}
\right.\]
where $s_i$ is the simple reflection $(i, i+1)$ in the order given by \eqref{eq:BasisD}. We have
\[\mathfrak{W}_{D_{n+1}} = \left\langle s_i^D \bigg| (s_i^Ds_j^D)^2 = \text{id} \ (|i - j| \ge 2, \{i, j\} \ne \{0, 2\}), (s_i^Ds_{i-1}^D)^3 = \text{id} \ (2 \le i \le n), (s_0^Ds_1^D)^2 = \text{id}, (s^D_0s^D_2)^3 = \text{id} \right\rangle,\]
and $\mathfrak{W}_{D_{n+1}} \cong \mathfrak{S}_{n+1} \ltimes \mathfrak{N}_D$ with $\mathfrak{N}_D \cong (\mathbb{Z}_2)^n$. Thus $|\mathfrak{W}_{D_{n+1}}| = 2^n(n+1)!$.

In $\mathfrak{h}^*_D \cong \mathbb{R}^{n+1}$, the $\mathfrak{S}_{n+1}$ part of $\mathfrak{W}_{D_{n+1}}$ is generated by the orthogonal reflections $s_{\varepsilon_i - \varepsilon_j} (i \ne j)$ which interchanges $\varepsilon_i$ and $\varepsilon_j$ and leaves invariant the $\varepsilon_k$ with $k \ne i, j$, while the normal part $\mathfrak{N}_D$ is generated by $s_{ij} = s_{\varepsilon_i - \varepsilon_j}s_{\varepsilon_i + \varepsilon_j}$ which transforms $\varepsilon_i$ to $-\varepsilon_i$, $\varepsilon_j$ to $-\varepsilon_j$ and leaves invariant the $\varepsilon_k$ with $k \ne i, j$. The generator $s_0^D$ we chose above has the effect of interchanging $\varepsilon_1$ and $\varepsilon_{0}$ and then transforming both into their negatives.

Note that we have a one-to-one correspondence between even length strict partitions and elements in the normal subgroup $\mathfrak{N}_D$ by sending $\lambda = (\lambda_1, \lambda_2, \dots, \lambda_{2k}) (n \ge \lambda_j \ge \lambda_{j+1} \ge 0)$ to the unique element $\tilde{w} \in \mathfrak{N}_D$ such that $\tilde{w}(\varepsilon_{\lambda_j}) = -\varepsilon_{\lambda_j}$ for $1 \le j \le 2k$ and $\tilde{w}$ does not change the sign of $\varepsilon_l$ for $l \in \{0, 1, \dots, n\}\backslash \{\lambda_1, \lambda_2, \dots, \lambda_{2k}\}$. We denote this Weyl group element as $\tilde{w}_{\lambda}$ in the following.

\subsection{$B_n$ as sub-Lie algebra of $D_{n+1}$}\label{sec:BnDn1}
Now $B_n$ can be realized inside $(2n + 2) \times (2n + 2)$ matrix algebra in the following way (where the index is the same as we used for $\mathfrak{so}(2n+2)$).
We take the following Chevalley generators:
\[\begin{array}{rcl}
H_{\alpha_i} & = & {H}_{\tilde\alpha_i} = \tilde{E}_{-i, -i} - \tilde{E}_{-(i-1), -(i-1)} + \tilde{E}_{+(i-1), +(i-1)} - \tilde{E}_{+i, +i}, \qquad 1 < i \le n,\\
H_{\alpha_1} & = &{H}_{\tilde\alpha_1} + {H}_{\tilde\alpha_0} = 2(\tilde{E}_{-1, -1} - \tilde{E}_{1, 1}),
\end{array}\]
where $\alpha_n = \varepsilon_n - \varepsilon_{n-1}, \dots, \alpha_2 = \varepsilon_2 - \varepsilon_1, \alpha_1 = \varepsilon_1$.

We take the simple root vectors,
\[\begin{array}{rcl}
X_{\alpha_i} & = & = {X}_{\tilde\alpha_i} = \tilde{E}_{-i, -(i - 1)} + \tilde{E}_{+(i - 1), +i} \qquad 1 < i \le n,\\
X_{\alpha_1} & = & = \frac{1}{\sqrt{2}}({X}_{\tilde\alpha_1} + {X}_{\tilde\alpha_0}) = \frac{1}{\sqrt{2}}(\tilde{E}_{-1, +0} + \tilde{E}_{-0, +1} + \tilde{E}_{-1, -0} + \tilde{E}_{+0, +1}),
\end{array}\]
and the negative root generators,
\[\begin{array}{rcl}
Y_{\alpha_i} & = & X_{-\alpha_i} = \tilde{E}_{-(i-1), -i} + \tilde{E}_{+i, +(i-1)} \qquad 1 < i \le n,\\
Y_{\alpha_1} & = & X_{-\alpha_1} =\sqrt{2} (\tilde{E}_{+0, -1} + \tilde{E}_{+1, -0} + \tilde{E}_{-0, -1} + \tilde{E}_{+1, +0}).
\end{array}\]
The other negative root vectors are then generated by taking the commutators as before
\[[Y_{\alpha}, Y_{\beta}] = N_{\alpha, \beta}Y_{\alpha + \beta} \qquad \text{if} \quad \alpha, \beta, \alpha + \beta \in \Sigma_+,\]
where the set of positive roots is given by
\[\Sigma_+ = \left\{\begin{array}{cl}
\alpha_1 + \cdots + \alpha_i & (1 \le i \le n),\\
2\alpha_1 + 2\alpha_2 + \cdots + 2\alpha_{i} + \alpha_{i+1} + \cdots + \alpha_j  \qquad & (1 \le i < j \le n),\\
\alpha_i + \cdots + \alpha_{j-1} & (1 \le i < j \le n). 
\end{array}
\right\}\]
We have $|\Sigma_+| = n^2$.

The Weyl group $\mathfrak{W}_{B_n}$ associated with the algebra of type $B_n$ can be given as follows
\[s_i^B = \left\{\begin{array}{ll}
{s}_{-i}{s}_{i-1}, \qquad & 1 < i \le n,\\
{s}_{-1}{s}_0{s}_{-1}, & i = 1,
\end{array}
\right.\]
where ${s}_i$ is the simple reflection $(i, i+1)$. We have
\[\mathfrak{W}_{B_n} = \left\langle s_i^B \bigg| (s_i^Bs_j^B)^2 = e \ (|i - j| \ge 2), (s_i^Bs_{i-1}^B)^3 = e \ (3 \le i \le n), (s_2^Bs_1^B)^4 = e \right\rangle,\]
and $\mathfrak{W}_{B_n} \cong \mathfrak{S}_n \ltimes \mathfrak{N}_B$ with $\mathfrak{N}_B \cong (\mathbb{Z}_2)^n$. Thus $|\mathfrak{W}_{B_n}| = 2^nn!$. 

Note that in $\mathfrak{h}^*_B \cong \mathbb{R}^n$, the $\mathfrak{S}_n$ part of $\mathfrak{W}_{B_n}$ is generated by orthogonal reflections $s_{\varepsilon_i - \varepsilon_j} (i \ne j)$ which interchanges $\varepsilon_i$ and $\varepsilon_j$ and leaves $\varepsilon_k$ invariant when $k \ne i, j$, while the normal part $\mathfrak{N}_B$ is generated by reflections $s_{\varepsilon_i}$ which transforms $\varepsilon_i$ to $-\varepsilon_i$ and leaves invariant the $\varepsilon_k$ for $k \ne i$. So we can associate a strict partition $\lambda = (\lambda_1, \lambda_2, \dots, \lambda_k)$ (i.e. $n \ge \lambda_j > \lambda_{j+1} \ge 1$) with the unique element $w \in \mathfrak{N}_B$ such that $w(\varepsilon_{\lambda_j}) = -\varepsilon_{\lambda_j}$ for $1 \le j \le k$ and $w$ stabilizes all the other short roots. We denote this Weyl group element as $w_{\lambda}$ in the following.

It can be checked that with the basis chosen as above, we indeed obtain $\mathfrak{so}_{2n+1}$ as a sub-Lie algebra of $\mathfrak{so}_{2n+2}$. Let $V \subset \tilde{V}$ be the subspace of $\tilde{V}$ spanned by
\[\Xi = (e_{-n}, e_{-n+1}, \dots, e_{-1}, e_0, e_1, \dots, e_n)\]
where $e_0 = \frac{1}{\sqrt{2}}(e_{-0} + e_{+0})$. Let $\Psi, Q$ be the restriction of $\tilde{\Psi}, \tilde{Q}$ on $V$, respectively, then $\mathfrak{so}_{2n+1} \subset \mathfrak{so}_{2n+2}$ is nothing but the orthogonal Lie algebra associated with $(V, Q)$.

\begin{remark}\label{rem:Weyliso}
We have a canonical isomorphism between $\mathfrak{N}_B$ and $\mathfrak{N}_D$: for strict partition $\lambda = (\lambda_1, \dots, \lambda_k)$, i.e. $\lambda_1 > \cdots \lambda_k > 0$, we have
\[\begin{array}{rcl}
\iota:  \mathfrak{N}_B & \longrightarrow & \mathfrak{N}_D\\
w_{\lambda} & \mapsto & \tilde{w}_{\lambda'}
\end{array}\]
where $\lambda' = \lambda$ if $k$ is even and $\lambda' = (\lambda_1, \dots, \lambda_k, 0)$ if $k$ is odd. With such identification $\mathfrak{W}_B$ can be viewed as a subalgebra of $\mathfrak{W}_D$ generated by $s_i^B=s_i^D \ (2 \le i \le n)$ and $s_1^B = s_0^Ds_1^D$.
\end{remark}

\begin{remark}
In the above realization of $\mathfrak{so}_{2n+1}$ in $(2n+2) \times (2n+2)$ matrices, elements in $\mathfrak{so}_{2n+1}$ are in the centralizer of the following $(2n+2) \times (2n + 2)$ matrix,
\begin{equation}\label{eq:Z}
\mathcal{Z} = \begin{pmatrix}[ccc|c|ccc]
 &    & &    &  & & \\
 & 0 & & 0 &  &0 &\\
 &  &  &   & & & \\
\hline 
& \raisebox{-0.05cm}{$ 0 $} & & z & &\raisebox{-0.05cm}{$ 0 $} & \\
\hline
 &    & &    &  & & \\
 & 0 & & 0 &  & 0  &\\
 &  &  &   & & & \\
\end{pmatrix},
\end{equation}
where $z$ is the $2 \times 2$ matrix in the form
\[z = \begin{pmatrix} 1 & -1 \\ -1 & 1 \end{pmatrix}.\]
Counting the dimensions, we find that this condition uniquely characterizes all the elements of $\mathfrak{so}_{2n+1}$ inside $\mathfrak{so}_{2n+2}$. That is, $X \in \mathfrak{so}_{2n+2}$ belongs to $\mathfrak{so}_{2n+1}$ if and only if
\[[X, \mathcal{Z}] = 0.\]
We can take an element in $\mathfrak{so}_{2n+2}$ with all the above properties by considering
\[\mathcal{Z}_0 = \mathcal{Z} - \frac{1}{n+1}I_{2n+2}.\]
In our study of BKP hierarchy, we may consider the following sequence of embeddings
\[\begin{array}{cccccccccc}
\mathfrak{so}_3 & \subset &  \mathfrak{so}_5 & \subset & \mathfrak{so}_7 & \subset &  \cdots &  \subset \mathfrak{so}_{2n+1} & \subset & \cdots\\
\cap & & \cap & & \cap & & & \cap & & \\
\mathfrak{so}_4 & \subset &  \mathfrak{so}_6 & \subset & \mathfrak{so}_8 & \subset & \cdots & \subset \mathfrak{so}_{2n+2} & \subset & \cdots
\end{array}\]
where $\mathfrak{so}_{2k+2}$'s are in their standard realization in $(2k+2) \times (2k+2)$ matrices. Since eventually we take the limit $n \to \infty$, we may embed the $(2n+2) \times (2n+2)$ matrices representing elements in $\mathfrak{so}_{2n+1}$ into the central part of $(2n+4) \times (2n+4)$ matrices, and replace $\mathcal{Z}$ by
\[\tilde{\mathcal{Z}} = \begin{pmatrix}
-1 & & \\
& \mathcal{Z} & \\
& & -1
\end{pmatrix}\in \mathfrak{so}_{2n+4}\]
before taking the relevant limit to make everything consistent.
\end{remark}


\subsection{The spin representation}\label{sec:Spin}

Let $\tilde{S}$ be the exterior algebra of the maximal isotropic subspace generated by $\{e_{-0}, e_{-1}, e_{-2}, \dots, e_{-n}\}$. It is well-known that the spin representation of $\mathfrak{so}_{2n+2}$ on $\tilde{S}$ is reducible and it decomposes into two irreducible $\mathfrak{so}_{2n+2}$-modules -- $\tilde{S}_+$ and $\tilde{S}_-$ which as vector space are linearly spanned by $\bigwedge^p \tilde{F}$ for $p$ even and odd, respectively (c.f. \cite{Chevalley1954}). These two representations are called semi-spinor representations of $\mathfrak{so}_{2n+2}$.

For example the spin representation of $\mathfrak{so}_{2n+2}$ on $\tilde{S}_-$ now takes the following form: First, the underlying vector space $\tilde{S}_-$ is given by
\[\tilde{S}_- = \{s \wedge e_{-0}: s \in S \text{ is even degree}\} \bigcup \{s : s \in S \text{ is odd degree}\},\]
where $S$ is the exterior algebra generated by $\{e_{-1}, e_{-2}, \dots, e_{-n}\}$. That is each element in $\tilde{S}_-$ can be uniquely written as
\[s = \sum \limits_{n \ge i_1 > i_2 > \cdots > i_k \ge 0}\xi_{i_1, i_2, \cdots, i_k}e_{-i_1} \wedge e_{-i_2} \wedge \cdots \wedge e_{-i_k}\]
where $k$ is an odd integer.
The action of Cl$(\tilde{V}, \tilde{Q})$ on $\tilde{S}_-$ is determined by
\begin{align*}
& \rho_-(x) \cdot (a_1 \wedge \cdots \wedge a_k) = x \wedge a_1 \wedge \cdots \wedge a_k, \qquad \text{for } x \in \tilde{F};\\
& \rho_-(y) \cdot (a_1 \wedge \cdots \wedge a_k) = \sum \limits_{i = 1}^k (-1)^{i-1}\tilde{\Psi}(a_i, y)a_1 \wedge \cdots \wedge a_{i-1} \wedge a_{i+1} \wedge \cdots \wedge a_k, \quad \text{for } y \in \tilde{F'}.
\end{align*}
At last the Lie algebra $\mathfrak{so}_{2n+2}(\mathbb{C})$ acts on $\tilde{S}_-$ through
\begin{align}\label{eq:spinD}
\begin{array}{ccc}
\varphi: \mathfrak{so}_{2n+2}(\mathbb{C}) & \longrightarrow & \text{Cl}(\tilde{V}, \tilde{Q})\\
(-1)^j\tilde{E}_{i, -j} - (-1)^i\tilde{E}_{j,-i} & \mapsto & \frac{1}{2}(e_ie_j-e_je_i).
\end{array}
\end{align}
We see that $e_{-n} \wedge e_{-n+1} \wedge \dots \wedge e_{-0}$ is the highest weight vector if $n$ is even, and $e_{-n} \wedge e_{-n+1} \wedge \dots \wedge e_{-1}$ is the highest weight vector if $n$ is odd. 

As a su-Lie balgebra of $\mathfrak{so}_{2n+2}$, $\mathfrak{so}_{2n+1}$ acts on $\tilde{S}_-$ through
\[ e_0 \cdot e_{-0} = \frac{1}{\sqrt{2}} \qquad \text{with} \quad e_0 = \frac{e_{+0} + e_{-0}}{\sqrt{2}}.\]
The spin representation of $\mathfrak{so}_{2n+1}$ on $\tilde{S}_+$ can be defined similarly with $e_{-n} \wedge e_{-n+1} \wedge \dots \wedge e_{-0}$ or $e_{-n} \wedge e_{-n+1} \wedge \dots \wedge e_{-1}$ as the highest weight vector when $n$ is odd or even, respectively.

Now we note that the two semi-spinor representations of $\mathfrak{so}_{2n+2}$ are actually equivalent representations when restricted to the sub-Lie algebra $\mathfrak{so}_{2n+1}$ of $\mathfrak{so}_{2n+2}$, and they are both isomorphic to the spin representation of $\mathfrak{so}_{2n+1}$ (c.f. \cite{Kodama-Xie2021f-KT}).  
The implications of this assertion to $\tau_{\text{BKT}}(\t_B; g)$ can be summarized in the following proposition.
\begin{proposition}
The square of the $\tau$-function defined through the spin representation in \eqref{eq:spinD} admits the following three expressions
\begin{align}\label{eq:tausquare}
\tau^2(\mathbf{t}_B; w_*,) & = \langle e_{-n} \wedge e_{-n+1} \wedge \cdots \wedge e_{-1} \wedge e_{-0}, \exp(\Theta_{e}(\mathbf{t}_B)) \tilde{N}\dot{w}_* \cdot e_{-n} \wedge e_{-n+1} \wedge \cdots \wedge e_{-1} \wedge e_{-0} \rangle \\
& = \langle e_{-n} \wedge e_{-n+1} \wedge \cdots \wedge e_{-1} \wedge e_{+0}, \exp(\Theta_{e}(\mathbf{t}_B)) \tilde{N}\dot{w}_* \cdot e_{-n} \wedge e_{-n+1} \wedge \cdots \wedge e_{-1} \wedge e_{+0} \rangle \nonumber\\
& = \langle e_{-n} \wedge e_{-n+1} \wedge \cdots \wedge e_{-1}, \exp(\Theta_{e}(\mathbf{t}_B)) \tilde{N}\dot{w}_* \cdot e_{-n} \wedge e_{-n+1} \wedge \cdots \wedge e_{-1}\rangle, \nonumber
\end{align}
where again $\langle \cdot, \cdot \rangle$ is the standard inner product on the wedge space $\bigwedge^{n+1}\mathbb{C}^{2n+2}$ or $\bigwedge^{n}\mathbb{C}^{2n+2}$. 
\end{proposition}


\subsection{The orthogonal Grassmannian manifold}\label{sec:Orthogonal}
Let $\mathfrak{so}_{2n+1}$ be the orthogonal Lie algebra associated with the quadratic space $(V, \Psi)$, where $V$ is a $2n+1$ dimensional vector space, and $\Psi$ a non-degenerate symmetric bilinear form of maximal index $n$ on $V$. A subspace $L \subset V$ is called Lagrangian if it is maximal totally isotropic with respect to $\Psi$. 
For example the space $L_0$  spanned by $e_i (i < 0)$ is a Lagrangian subspace of $V$. We let $\mathcal{P}_0 = \{T \in \mathcal{SO}(V) | TL_0 = L_0\} \subset \mathcal{SO}(V)$ be the isotropic subgroup of $L_0$, then the orthogonal Grassmannian manifold is defined as
\[\text{OG}(V)= \mathcal{SO}(V) \slash \mathcal{P}_0 = \{\text{the totality of Lagrangian subspaces of $V$}\}.\]
Since dim $\mathcal{SO}(V) = n(2n+1)$ and a maximal parabolic subgroup of $\mathcal{SO}(V)$ which fix $L_0$ has dimension $n^2+n+\frac{n(n-1)}{2}$, we have dim $\text{OG}(V) = \frac{n(n+1)}{2}$. The orthogonal Grassmannian manifold $\text{OG}(V)$ has the following Schubert cell decomposition
\[\mathcal{SO}(V) = \bigsqcup \limits_{w_{\lambda} \in \mathfrak{N}_B}\mathcal{N}_-\dot{w}_{\lambda}\mathcal{P}_0.\]
We will denote by $\mathcal{L}_{\lambda}$  the Schubert cell associated with $w_{\lambda} \in \mathfrak{N}_B$ in the following.

Let $(\rho, S)$ be the spin representation of $\mathfrak{so}_{2n+1}$ with highest weight vector $v^{\omega_1}$ as defined in Section \ref{sec:Spin}. 
Each element in $S$ can be uniquely written as
\begin{equation}\label{eq:spinor}
s = \sum \limits_{n \ge i_1 > i_2 > \cdots > i_k > 0} \xi_{i_1, i_2, \cdots, i_k}e_{-i_1} \wedge e_{-i_2} \wedge \cdots \wedge e_{-i_k}.
\end{equation}
Elements in $S$ are usually called spinors in the literature. E. Cartan showed that OG$(V)$ can  be embedded into $\mathbb{P}(S)$ through the so-called Cartan map as follows (c.f.  \cite{Cartan1966, Balogh-Harnad-Hurtubise2021}). Let $L = \mathbb{C}\langle w_1, w_2, \dots, w_n\rangle \in \text{OG}(V)$ be a point in the orthogonal Grassmannian manifold with basis $\{w_1, w_2, \dots, w_n\}$, then it can be checked that $(\rho(w_1) \cdot \dots \cdot (\rho(w_{n-1}) \cdot (\rho(w_n) \cdot S))) \subset S$ is a $1$ dimension dimensional subspace of $S$, and we define the image of $w$ in $\mathbb{P}(S)$ to be the unique point representing this one dimension subspace $\mathbb{C}\langle \rho(w_1) \cdot \dots \cdot (\rho(w_{n-1}) \cdot (\rho(w_n) \cdot S)) \rangle$. For example, taking $L = L_0$, then the one dimensional subspace of $S$ is generated by $e_{-1} \wedge \dots \wedge e_{-n} \in S$. The image of OG$(V)$ under the Cartan map is the orbit of $SO(V)$ through $e_{-1} \wedge \dots \wedge e_{-n} \in S$ in $\mathbb{P}(S)$ through the spin representation, and elements sitting in the image of the Cartan map are called pure spinors. Let $s \in S$ be a pure spinor, then the Lagrangian subspace $L_s$ associated with $s$ is the linear span of $w \in V$ such that
\[\rho(w) \cdot s = 0, \qquad \text{for } w \in L_s.\]
Thus, a pure spinor is an element in $S$ which represents a point in the orthogonal Grassmannian manifold. 

Since dim $\mathbb{P}(S) = 2^n - 1 \ge \text{dim OG}(V) = \frac{n(n+1)}{2}$, the image of OG$(V)$ in $\mathbb{P}(S)$ has a huge codimension when $n$ is large, and they are characterized by a system of complicated quadratic equations called Cartan-Pl\"ucker relations (c.f. \cite{Cartan1966, Chevalley1954}). We will give a constructive description for the Cartan coordinates of pure spinors in Section \ref{sec:purespinor}.


\section{Generic solutions of $B$-Toda and BKP hierarchies}\label{sec:genericsolution}
In this section, we give a quick derivation of the Pfaffian expression for $\tau$-functions of the $B$-Toda and the BKP hierarchies in the generic case, that is when $w_* = \text{id}$.

\subsection{Normal form for elements in the orthogonal Grassmannian manifold: the generic case}\label{sec:normalforme}
Recall that our goal is to simplify the expression for the square of the $\tau$-function defined by
\[\tau_{\text{BKT}}^2(\mathbf{t}_B; w_*)  = \langle e_{-n} \wedge e_{-n+1} \wedge \cdots \wedge e_{1} \wedge e_{-0}, \exp(\Theta_{e}(\mathbf{t}_B)) \tilde{N}\dot{w}_* \cdot e_{-n} \wedge e_{-n+1} \wedge \cdots \wedge e_{-1} \wedge e_{-0} \rangle,\]
and for that purpose we would like to find a convenient normal form for the unipotent matrix $\tilde{N}$. 

We deal with the case when $w_* = \text{id}$ in this section. That is we will calculate the principal $(n+1) \times (n+1)$ minor of the $(2n + 2) \times (2n + 2)$ matrix $\exp(\Theta_{e}(\mathbf{t}_B)) \tilde{N}$. 

First we note that the one-dimensional subspace spanned by the highest weight vector $e_{-n} \wedge e_{-n+1} \wedge \cdots \wedge e_{-1} \wedge e_{-0} $ is invariant under the Siegel parabolic subgroup $\mathcal{P}_0 \in \mathcal{SO}_{2n+1}(\mathbb{C})$, i.e. $(2n + 2) \times (2n + 2)$ matrices in the following form:
\[P_0 = \begin{pmatrix}[ccc|c|ccc]
 &    & &    &  & & \\
 & A_0 & & w_0 &  & B_0 &\\
 &  &  &   & & & \\
\hline 
& \raisebox{-0.05cm}{$ 0 $} & & I_2 & &\raisebox{-0.05cm}{$ z_0^T $} & \\
\hline
 &    & &    &  & & \\
 & 0 & & 0 &  & D_0  &\\
 &  &  &   & & & \\
\end{pmatrix}\]
where $A_0, B_0, D_0$ are $n \times n$ matrices and $w_0, z_0$ are $n \times 2$ matrices.

Then our goal is to choose a convenient representative in the class $\tilde{N}\mathcal{P}_0$. 

\begin{proposition}\label{prop:normalforme}
There is a unique $N_{\text{id}} \in \tilde{N}\mathcal{P}_0$ in the following form:
\begin{equation}\label{eq:normalformn}
N_{\text{id}} = \begin{pmatrix}[ccc|c|ccc]
 &    & &    &  & & \\
 & J & & 0 &  & 0 &\\
 &  &  &   & & & \\
\hline 
& \raisebox{-0.05cm}{$ -r^T $} & & I_2 & &\raisebox{-0.05cm}{$ 0 $} & \\
\hline
 &    & &    &  & & \\
 & R - \frac{1}{2}rr^T & & r &  & J  &\\
 &  &  &   & & & \\
\end{pmatrix},
\end{equation}
where $J$ is an $n \times n$ matrix defined in Equation \eqref{eq:tildeM}, $R = (r_{i,j})_{1 \le i, j \le n}$ is an $n \times n$ skew-symmetric matrix and the transpose of $r$ is a $2 \times n$ matrix in the following form:
\[r^T = \begin{pmatrix}
r_{1,0} & r_{2,0} & \cdots & r_{n,0}\\
r_{1,0} & r_{2,0} & \cdots & r_{n,0}
\end{pmatrix}.\]
\end{proposition}

\begin{proof}
First we note that as an element in $\mathcal{SO}_{2n+2}(\mathbb{C})$, $P_0$ is constrained by
\begin{align*}
& P_0^T\tilde{M}P_0 = \tilde{M}.
\end{align*}
Writing this down explicitly, we have
\begin{align*}
& A_0^T J D_0 = J,\\
& D_0^TJ^Tw_0 + z_0\sigma = 0,\\
& D_0^TJ^TB_0 + z_0\sigma z_0^T + B_0^T J D_0 = 0,
\end{align*}
where $\sigma$ is the $2 \times 2$ matrix defined in equation \eqref{eq:tildeM}.

Thus $A_0 \in \mathcal{GL}_n(\mathbb{C})$ and the $2 \times n$ matrix $w_0$ could be arbitrarily chosen, $B_0$ is an $n \times n$ matrix such that
\[B_0^TJD_0 + \frac{1}{2}z_0\sigma z_0^T\]
is a skew-symmetric and $D_0, z_0$ could then be determined by the other equalities above subsequently. Of course as an element in $\mathcal{SO}_{2n+1}(\mathbb{C})$ we need to put some further constraints on $w_0$, but we do not need them for now.

On the other hand, a general element in $\mathcal{N}_-$ has the following form
\[\tilde{N} = \begin{pmatrix}[ccc|c|ccc]
 &    & &    &  & & \\
 & \tilde{A} & & 0 &  & 0 &\\
 &  &  &   & & & \\
\hline 
& \raisebox{-0.05cm}{$ \tilde{r}^T $} & & I_2 & &\raisebox{-0.05cm}{$ 0 $} & \\
\hline
 &    & &    &  & & \\
 & \tilde{C} & & \tilde{s} &  & \tilde{D}  &\\
 &  &  &   & & & \\
\end{pmatrix} \qquad \text{with} \begin{array}{l}
\tilde{C} = (\tilde{A}^TJ)^{-1}(\tilde{R} - \frac{1}{2}\tilde{r}\tilde{r}^T) \text{ where $\tilde{R}$ is skew-symmetric} ,\\
\tilde{A}^TJ\tilde{D} = J \text{ and } \tilde{A}^TJ\tilde{s} + \tilde{r} = 0.
\end{array}
\]
The first column of $\tilde{N}P_0$ in block form then reads as:
\begin{align*}
\begin{pmatrix}[ccc]
 &    &  \\
 & \tilde{A}A_0 & \\
 &  &   \\
\hline 
& \raisebox{-0.05cm}{$\tilde{r}^TA_0$} &  \\
\hline
 &    &  \\
 & \tilde{C}A_0 & \\
 &  &   \\
\end{pmatrix}.
\end{align*}
Now we specialize our choice of $A_0$ to simplify the above matrix by taking $A_0 = \tilde{A}^{-1}J$, and let
\begin{align*}
& r = A_0^T\tilde{r}, \qquad \text{or equivalently} \quad r^T = \tilde{r}^T\tilde{A}^{-1}J,\\
& R = \tilde{C}A_0 + \frac{1}{2}rr^T = \tilde{C}\tilde{A}^{-1}J + \frac{1}{2}rr^T.
\end{align*}
This partially constraints the entries of $P_0$. We choose an arbitrary one satisfying these constraints and denote it by $P_{\tilde{N}}$ . Then the condition $N_{\text{id}} = \tilde{N}P_{\tilde{N}} \in \mathcal{SO}_{2n+2}(\mathbb{C})$, i.e. $N_{\text{id}}^T \tilde{M}N_{\text{id}} = \tilde{M}$, implies that it has the form \eqref{eq:normalformn} with $R$ an $n \times n$ skew-symmetrix matrix, and the condition $N_{\text{id}} \in \mathcal{SO}_{2n+1}$ further implies that $r$ has the form in the statement of the Proposition. Counting dimensions, we see that other than these constraints, $R$ and $r$ are arbitrary. 
\end{proof}

\begin{remark}\label{rem:normalforme}
Note that
\begin{equation*} 
\begin{pmatrix}[ccccc]
0 & 0 & \cdots & (-1)^n & 0 \\
\vdots & \vdots & \udots & \vdots & 0 \\
0 & 1 & \cdots & \cdots & 0 \\
-1 & 0 & \cdots & \vdots & 0 \\
-r_{10} & -r_{20} & \cdots & -r_{n0} & 1\\
-r_{10} & -r_{20} & \cdots & -r_{n0} & 0\\
-r_{10}^2 & -r_{21} - r_{10}r_{20} & \cdots & -r_{n1} - r_{10}r_{n0} & r_{10}\\
r_{21} - r_{10}r_{20} & - r_{20}^2 & \cdots & -r_{n2} - r_{20}r_{n0} & r_{20}\\
\vdots & \vdots & \ddots & \vdots & \vdots\\
r_{n1} - r_{10}r_{n0} & r_{n2} - r_{20}r_{n0} & \cdots & -r_{n0}^2 & r_{n0} 
\end{pmatrix} = \begin{pmatrix}[ccccc]
0 & 0 & 0 & \cdots & (-1)^n  \\
\vdots & \vdots & \vdots & \udots & \vdots  \\
0 & 0 & 1 & \cdots & \cdots  \\
0 & -1 & 0 & \cdots & \vdots  \\
1 & 0 & 0 & \cdots & 0 \\
0 & -r_{10} & -r_{20} & \cdots & -r_{n0} \\
r_{10} & 0 & -r_{21}  & \cdots & -r_{n1} \\
r_{20} & r_{21}  & 0 & \cdots & -r_{n2}\\
\vdots & \vdots & \ddots & \vdots & \vdots\\
r_{n0} & r_{n1}  & r_{n2} & \cdots & 0 
\end{pmatrix} \begin{pmatrix}[ccccc]
-r_{10} & -r_{20} & \cdots & -r_{n0} & 1\\
1 & 0 & \cdots & \vdots & \vdots\\
0 & 1 & \cdots & \vdots & \vdots \\
\vdots & \vdots & \ddots & \vdots & \vdots\\
0 &  0 & \cdots & 1 & 0 
\end{pmatrix}.
\end{equation*}
We denote the two matrices on the right hand side by $G_{\tilde{N}}$ and $A_{\tilde{N}}$ respectively in the following, then we have 
\begin{align*}
& \tilde{N} \cdot e_{-n} \wedge e_{-n+1} \wedge \dots \wedge e_{-1} \wedge e_{-0} \\
&  = \sum \limits_{I \in \binom{[-n, n]}{n+1}}\Delta_{I, [-n_0]}(\tilde{N}) e_{i_1} \wedge e_{i_2} \wedge \dots \wedge e_{i_{n+1}}\\
&  = d_{\tilde{N}}\sum \limits_{I \in \binom{[-n, n]}{n+1}}\Delta_{I, [-n_0]}({N}_{\text{id}}) e_{i_1} \wedge e_{i_2} \wedge \dots \wedge e_{i_{n+1}}\\
&  = -d_{\tilde{N}} \sum \limits_{I \in \binom{[-n, n]}{n+1}}\Delta_I({G_{\tilde{N}}}) e_{i_1} \wedge e_{i_2} \wedge \dots \wedge e_{i_{n+1}},
\end{align*}
where $d_{\tilde{N}} = \det(J\tilde{A})$, $\Delta_I(G_{\tilde{N}})$ is the minor of $G_{\tilde{N}}$ with rows indexed by $I$, and $I = {i_1 \le i_2 \dots \le {i_{n+1}}} \subset [-n, n] := \{-n, -n+1, \dots, -1, -0, +0, 1, \dots, n\}$,  $[-n_0] = \{-n, -n+1, \dots, -1, -0\}$ are two index sets.
\end{remark}

\subsection{Schur-$Q$ function}\label{sec:SchurQ}
It is known that Schur-$Q$ functions are basis of the polynomial solutions of the BKP hierarchy (see, e.g. \cite{You1989}). Let's give a brief introduction to these symmetric functions\footnote{The Schur-Q function as we defined here is a little bit different from what normally used in the literature for combinatorics (see, e.g. \cite{Jozefiak-Pragacz1991, Macdonald1995}). $Q_{\lambda}(\t_B)$ is actually the function $g_{\lambda}(\t_B)$ introduced in the BKP setting in \cite{You1989}. Denoting the functions defined in \cite{Jozefiak-Pragacz1991, Macdonald1995} by $\tilde{Q}_{\lambda}$, then $Q_{\lambda}(\ttt_B) = 2^{-l(\lambda) \slash 2}\tilde{Q}_{\lambda}(\ttt_B \slash 2)$, where $l(\lambda)$ is the length of $\lambda$.}, and see how they arise naturally in our context.

Let $q_i(\ttt_B)$ be polynomials in $\t_B = (t_1, t_3, t_5, \cdots)$ defined by the following generating series
\[e^{\theta(\ttt_B, z)} = \sum \limits_{i = 0}^{\infty}q_i(\ttt_B)z^i \qquad \text{with} \quad \theta(\ttt_B, z) = \sum \limits_{k = \text{odd}}t_kz^k.\]
For integers $\lambda_i, \lambda_j \ge 0$ we define the elementary Schur-$Q$ functions to be
\[Q_{\lambda_i, \lambda_j}(\ttt_B) = \frac{1}{2} q_{\lambda_i}(\ttt_B)q_{\lambda_j}(\ttt_B) +  \sum \limits_{k = 1}^{\lambda_j}(-1)^k q_{\lambda_i + k}q_{\lambda_j - k}.\]
Note that $Q_{\lambda_i, \lambda_j} = -Q_{\lambda_j, \lambda_i}$ and $Q_{\lambda_1,0} = \frac{1}{\sqrt{2}}q_{\lambda_1} = -Q_{0, \lambda_1}$. For each strict partition $\lambda = (\lambda_1, \dots, \lambda_m)$, 
we define $Q_{\lambda}$ as follows
\begin{equation}
Q_{\lambda} = \left\{\begin{array}{lll}
\text{Pf}(Q_{\lambda_i,\lambda_j}), \qquad & 1 \le i, j \le m \qquad & \text{for $m$ even},\\
\text{Pf}(Q_{\lambda'_i,\lambda'_j}), \qquad & 1 \le i, j \le m + 1 \qquad & \text{for $m$ odd},
\end{array}
\right.
\end{equation}
where $\lambda' = (\lambda_1, \dots, \lambda_m, 0)$ and Pf means the Pfaffian of the indicated skew-symmetric matrix.


\subsection{$\tau$-functions of the BKP hierarchy in terms of Pfaffian: the generic case}\label{sec:Pfaffiangeneric}

Now we apply Proposition \ref{prop:normalforme} and Remark \ref{rem:normalforme} to the following $(2n+2) \times (2n+2)$ matrix $\mathcal{P}_{2n+2} \in \mathcal{SO}_{2n+1}$:
\[
\mathcal{P}_{2n+2} = \exp(\Theta_{f}(\mathbf{t}_B)) = \begin{pmatrix}[ccc|c|ccc]
 &    & &    &  & & \\
 & \tilde{A} & & 0 &  & 0 &\\
 &  &  &   & & & \\
\hline 
& \raisebox{-0.05cm}{$\tilde{r}^T$} & & I_2 & &\raisebox{-0.05cm}{$ 0 $} & \\
\hline
 &    & &    &  & & \\
 & \tilde{C} & & \tilde{s} &  & \tilde{D}  &\\
 &  &  &   & & & \\
\end{pmatrix},
\]
where $f = e^T = (\sum \limits_{i = 1}^n X_{\alpha_i})^T$ with
\[\tilde{A} = \tilde{D} = \begin{pmatrix}
1 & & & \\
q_1 & 1 & & \\
\vdots & \vdots & \ddots & \\
q_{n-1} & q_{n-2} & \cdots & 1
\end{pmatrix}, \quad \tilde{r} = \begin{pmatrix} \frac{1}{\sqrt{2}}q_n  & \frac{1}{\sqrt{2}} q_n \\ \frac{1}{\sqrt{2}}q_{n-1} & \frac{1}{\sqrt{2}}q_{n-1} \\ \vdots & \vdots \\ \frac{1}{\sqrt{2}}q_1 & \frac{1}{\sqrt{2}}q_1 \end{pmatrix},  \quad \tilde{s} = \begin{pmatrix}  \frac{1}{\sqrt{2}}q_1 & \frac{1}{\sqrt{2}} q_1 \\ \frac{1}{\sqrt{2}}q_2 & \frac{1}{\sqrt{2}}q_2 \\ \vdots \\ \frac{1}{\sqrt{2}}q_n & \frac{1}{\sqrt{2}}q_n \end{pmatrix}, \quad \tilde{C} = \begin{pmatrix}
q_{n+1} & \cdots & q_2\\
\vdots & & \vdots\\
q_{2n} & \cdots & q_{n+1}
\end{pmatrix},
\]
and $q_i = q_i(\ttt_B)$.

Then
\[\tilde{A}^{-1}J = \begin{pmatrix}[rrrccc]
& & & & & (-1)^n\\
& & & & (-1)^{n-1} & (-1)^{n-1}q_1\\
& & & \udots & \vdots & \vdots \\
& & -1 & \cdots & -q_{n-4} & -q_{n-3}\\
& 1 & q_1 & \cdots & q_{n-3} & q_{n-2} \\
-1 & -q_1 & -q_2 & \cdots & -q_{n-2} & -q_{n-1}
\end{pmatrix}\]
and using the identities
\[q_k^2 + 2 \sum \limits_{i = 1}^k(-1)^iq_{k-i}q_{k+i} = 0 \qquad k \ge 1,\]
we obtain
\[r^T = \tilde{r}^T\tilde{A}^{-1}J = \begin{pmatrix} -\frac{1}{\sqrt{2}}q_1 &  -\frac{1}{\sqrt{2}}q_2 & \cdots &  -\frac{1}{\sqrt{2}}q_{n}\\
-\frac{1}{\sqrt{2}}q_1 &  -\frac{1}{\sqrt{2}}q_2 & \cdots &  -\frac{1}{\sqrt{2}}q_{n}
\end{pmatrix}\]
and
\[Q = \tilde{C}\tilde{A}^{-1}J + \frac{1}{2}rr^T= \begin{pmatrix}
0 & -Q_{2,1} & -Q_{3,1} & \cdots & -Q_{n, 1}\\
Q_{2, 1} & 0 & -Q_{3, 2} & \cdots & -Q_{n, 2}\\
Q_{3, 1} & Q_{3, 2} & 0 & \cdots & \vdots\\
\vdots & \vdots & \vdots & \ddots & \vdots\\
Q_{n, 1} & Q_{n, 2} & Q_{n, 3} & \cdots & 0 
\end{pmatrix},\]
where $Q_{\lambda_i, \lambda_j} = Q_{\lambda_i, \lambda_j}(\ttt_B)$ is exactly the elementary Schur-$Q$ function defined in Section \ref{sec:SchurQ}.

Now, by Remark \ref{rem:normalforme}, we have
\begin{align*}
\bar{\mathcal{P}}_{2n+2} & = \exp(\Theta_{f}(\mathbf{t}_B)) \cdot e_{-n} \wedge e_{-n+1} \wedge \cdots \wedge e_{-1} \wedge e_{-0}\\
& = -d_{\mathcal{P}_{2n+2}} \sum \limits_{I \in \binom{[-n, n]}{n+1}}\Delta_I({G_{\mathcal{P}_{2n+2}}}) e_{i_1} \wedge e_{i_2} \wedge \dots \wedge e_{i_{n+1}},
\end{align*}
where
\[G_{\mathcal{P}_{2n+2}} = \begin{pmatrix}[ccccc]
0 & 0 & 0 & \cdots & (-1)^n  \\
\vdots & \vdots & \vdots & \udots & \vdots  \\
0 & 0 & 1 & \cdots & \cdots  \\
0 & -1 & 0 & \cdots & \vdots  \\
1 & 0 & 0 & \cdots & 0 \\
0 & Q_{0, 1} & Q_{0,2} & \cdots & Q_{0,n} \\
Q_{1,0} & 0 & Q_{1, 2}  & \cdots & Q_{1,n} \\
Q_{2,0} & Q_{2, 1}  & 0 & \cdots & Q_{2,n}\\
\vdots & \vdots & \vdots & \ddots & \vdots\\
Q_{n,0} & Q_{n, 1}  & Q_{n, 2} & \cdots & 0 
\end{pmatrix}.\]

By \eqref{eq:tausquare}, when $w_* = \text{id}$ we have that
\begin{align*}
\tau_{\text{BKT}}^2(\t_B) & = \langle \exp(\Theta_{f}(\mathbf{t}_B)) \cdot e_{-n} \wedge e_{-n+1} \wedge \cdots \wedge e_{-1} \wedge e_{-0}, \tilde{N} \cdot e_{-n} \wedge e_{-n+1} \wedge \cdots \wedge e_{-1} \wedge e_{-0} \rangle\\
& = \det(G_{\mathcal{P}_{2n+2}}^TG_{\tilde{N}})\\
& = \det
\begin{pmatrix}[ccc|ccc]
 &    &    &  & & \\
 & Q_{n+1} &   &  & \tilde{J}_{n+1} &\\
 &  &     & & & \\
\hline 
 &    & &      & & \\
 & -\tilde{J}^T_{n+1} & &  & R_{n+1}  &\\
 &  &     & & & \\
\end{pmatrix},
\end{align*}
where $\tilde{J}_{n + 1}$ is the following $(n+1) \times (n+1)$ matrix
\[\tilde{J}_{n + 1} = \begin{pmatrix}
& & & 1\\
& & \udots &\\
& 1 & & \\
1 & & &
\end{pmatrix}, 
\]
and
\[Q_{n+1} = \begin{pmatrix}
0 & \cdots & -Q_{n,2} & -Q_{n,1} & -Q_{n,0}\\
\vdots & \ddots & \vdots & \vdots & \vdots\\
Q_{n, 2} & \cdots & 0 & -Q_{2,1} & -Q_{2,0}\\
Q_{n, 1} & \cdots & Q_{2,1} & 0 & -Q_{1,0}\\
Q_{n, 0} & \cdots &  Q_{2,0} & Q_{1,0} & 0
\end{pmatrix}, 
\qquad R_{n+1} = \begin{pmatrix}
0 & -r_{1, 0} & -r_{2,0} & \cdots & -r_{n,0}\\
r_{1,0} & 0 & -r_{2,1} & \cdots & -r_{n,1}\\
r_{2,0} & r_{2,1} & 0 & \cdots & -r_{n,2}\\
\vdots & \vdots & \vdots & \ddots & \vdots\\
r_{n,0} & r_{n,1} & r_{n,2} & \cdots & 0
\end{pmatrix}.
\]
Here we have used the following well-known equality from linear algebra
\[\det(I + AB) = \det(I + BA) = \det
\begin{pmatrix}[ccc|ccc]
 &    &    &  & & \\
 & A &   &  & \tilde{J}_{n+1} &\\
 &  &     & & & \\
\hline 
 &    & &      & & \\
 & -\tilde{J}^T_{n+1} & &  & B  &\\
 &  &     & & & \\
\end{pmatrix}.\]
Since $\tau_{\text{BKT}}^2(\t_B)$ is the determinant of a $(2n+2) \times (2n+2)$ skew-symmetric matrix, we can take $\tau_{\text{BKT}}(\t_B)$ as the Pfaffian of the same matrix, that is
\begin{align}\label{eq:Skew2n2}
\tau_{\text{BKT}}(\t_B) & = \text{Pf}\begin{pmatrix}[ccc|ccc]
 &    &    &  & & \\
 & Q_{n+1} &   &  & \tilde{J}_{n+1} &\\
 &  &     & & & \\
\hline 
 &    & &      & & \\
 & -\tilde{J}^T_{n+1} & &  & R_{n+1}  &\\
 &  &     & & & \\
\end{pmatrix}\\
& 
=  \sum \limits_{\lambda \in \text{DP}}\text{Pf}(R_{\lambda})Q_{\lambda}(\t_B). \nonumber
\end{align}

Let $n \to \infty$, we obtain Theorem \ref{thm:taufunctiongeneric}.

\begin{remark}
We note that in cases when combinatorics is more important, that is when we care more about the structure of the coefficients in the Schur-$Q$ expansion rather than the convergence of $\tau$-functions, then unlike in the KP case, we do not need to put any finiteness constraints on the skew-symmetric matrix $R$ in \eqref{eq:Pfaffiantau} since any coefficient in the Schur-$Q$ expansion of $\tau_{\text{BKP}}(\t_B)$ is the Pfaffian of a certain finite dimensional skew-symmetric matrix which is automatically well-defined.
\end{remark}


\section{Parameterization of the orthogonal Grassmannian manifolds and general solutions of the $B$-Toda and the BKP hierarchies}\label{sec:generalcase}

\subsection{Schubert strata of the orthogonal Grassmannian manifolds: a review of the generic case}
First we note that $\tau_{\text{BKT}}^2(\t_B)$ is invariant under the action of the $\mathfrak{S}_B \subset \mathfrak{S}_D$ part of the Weyl groups $\mathfrak{W}_B \subset \mathfrak{W}_D$ (see Section \ref{sec:BninDn1}). To construct the general solutions of the $B$-Toda and BKP hierarchies, we need to classify all the Lagrangian subspaces of $(\tilde{V}, \tilde{Q})$ for $\mathfrak{so}_{2n+1}$. That is, for each element $w_{\lambda} \in \mathfrak{N}_B$ in the Weyl group of $\mathfrak{so}_{2n+1}$ we would like to find a normal form for elements in the Schubert cell $\mathcal{L}_{\lambda}$ of the orthogonal Grassmannian as we did in Section \ref{sec:genericsolution} for $w_{\emptyset} = \text{id}$.

In this section we show that the same procedure in Section \ref{sec:normalforme} could be applied to the case when $w \ne \text{id}$. Here instead of $\mathcal{P}_0$ which keeps the one-dimensional subspace spanned by $e_{-n} \wedge e_{-n+1} \wedge \cdots \wedge e_{-1} \wedge e_{-0}$ invariant, we consider parabolic subgroups of $\mathcal{SO}_{2n+1}(\mathbb{C})$ which fix the vector $\dot{w}_*\cdot e_{-n} \wedge e_{-n+1} \wedge \cdots \wedge e_{-1} \wedge e_{-0}$ up to scalars, and they are given by $\mathcal{P}_{w_*} = \dot{w}_* \mathcal{P}_0 (\dot{w}_*)^{-1}$. Then in analogous to Proposition \ref{prop:normalforme}, we would like to find a canonical form $N_{w_*} \in \tilde{N}P_{w_*}$, which simplifies the following expression
\begin{equation}\label{eq:BKTtau}
\tau_{\text{BKT}}^2(\t_B; w_*) = \langle \exp(\Theta_f(\t_B))e_{-n} \wedge e_{-n+1} \wedge \cdots \wedge e_{-1} \wedge e_{-0}, \tilde{N}\dot{w}_* \cdot e_{-n} \wedge e_{-n+1} \wedge \cdots \wedge e_{-1} \wedge e_{-0}\rangle.
\end{equation}

Since we have already obtained a satisfactory expression for $\langle\exp(\Theta_f(\t_B))e_{-n} \wedge e_{-n+1} \wedge \cdots \wedge e_{-1} \wedge e_{-0}$ in Section \ref{sec:Pfaffiangeneric}, we will focus on the expression $\tilde{N}\dot{w}_* \cdot e_{-n} \wedge e_{-n+1} \wedge \cdots \wedge e_{-1} \wedge e_{-0}\rangle$ in the following.

First let us summarize the procedure we used to find the orthogonal Grassmannian manifold corresponding to $w_* = \text{id}$ in a way suitable for generalization to other nontrivial Weyl group elements. The canonical element $N \in \tilde{N}\mathcal{P}_0$ we choose in Proposition \ref{prop:normalforme} is actually not lower unipotent, but it has the advantage that the lower left corner of ${N}_{\text{id}}$ is almost skew-symmetric, and its skew-symmetrization, the matrix $R$ serves as the inhomogeneous coordinate for the big Schubert cell in the orthogonal Grassmannian manifold. Equipped with this observation, the relation between $N_{\text{id}}$ and a normalized lower unipotent element $N_0$ for $\tilde{N}\mathcal{P}_0$ is almost immediate to obtain. We have
\begin{align*}
& L_{\text{id}} = N_{\text{id}}T = N_0\tilde{M} = \begin{pmatrix}[ccc|c|ccc]
 &    & &    &  & & \\
 & 0 & & 0 &  & J &\\
 &  &  &   & & & \\
\hline 
& \raisebox{-0.05cm}{$ 0 $} & & \sigma & &\raisebox{-0.05cm}{$ -r^T $} & \\
\hline
 &    & &    &  & & \\
 & J^T & & r &  & R - \frac{1}{2}rr^T  &\\
 &  &  &   & & & \\
\end{pmatrix},
\end{align*}
where $\tilde{M}$ is the Gram matrix for $\mathfrak{so}_{2n+2}(\mathbb{C})$ (see Section \ref{sec:in2n+2}), and
\[T = \begin{pmatrix}[ccc|c|ccc]
 &    & &    &  & & \\
 & 0 & & 0 &  & I_n &\\
 &  &  &   & & & \\
\hline 
& \raisebox{-0.05cm}{$ 0 $} & & \sigma & &\raisebox{-0.05cm}{$ 0 $} & \\
\hline
 &    & &    &  & & \\
 & (J^T)^2 & & 0 &  & 0  &\\
 &  &  &   & & & \\
\end{pmatrix} \qquad \text{and} \qquad 
N_0 = \begin{pmatrix}[ccc|c|ccc]
 &    & &    &  & & \\
 & I_n & & 0 &  & 0 &\\
 &  &  &   & & & \\
\hline 
& \raisebox{-0.05cm}{$ \tilde{r}^T\tilde{A}^{-1} $} & & I_2 & &\raisebox{-0.05cm}{$ 0 $} & \\
\hline
 &    & &    &  & & \\
 & (\tilde{A}^{-1}J)^{T}(\tilde{R} - \frac{1}{2}\tilde{r}\tilde{r}^T)\tilde{A}^{-1} & & -(\tilde{A}^{-1}J)^{T}\tilde{r} &  & I_n  &\\
 &  &  &   & & & \\
\end{pmatrix}.
\]

We note that the last $(n+1)$-columns of $N_{\text{id}}T$ consists of a matrix representative for the Lagrangian subspace of $(\tilde{V}, \tilde{Q})$. And with proper column operation we obtain the matrix $R_{n+1}$ appearing in the lower right corner of $\tau_{\text{BKT}}(\t_B)$ in equation \eqref{eq:Skew2n2}. More precisely, the column operation can be achieved by the following $(2n+2) \times (2n+2)$ matrix $\tilde{U}_{\text{id}}$ which is an element in $\mathcal{O}(2n+2)$, but not an element in $\mathcal{O}(2n+1)$:
\[\tilde{U}_{\text{id}} = \begin{pmatrix}
1 &  &  &   & (-1)^{n-1}r_{n, 0} & & & & & \\
& \ddots & &  & \vdots & & & & & \\
& & 1 &  & -r_{2, 0} & & & & & \\
& & & 1 & r_{1, 0} & & & & & \\
& & & & 1 & & & & & \\
& & & & & 1 & r_{1, 0} & r_{2, 0} & \cdots & r_{n, 0}\\
& & & & & & 1 & & &\\
& & & & & & & 1 & &\\
& & & & & & & & \ddots &\\
& & & & & & & & & 1  
\end{pmatrix}.\]
Then the nontrivial part of the $(2n+2) \times (2n+2)$ matrix $N_{\text{id}}T\tilde{U}_{\text{id}}$ is the normal form appearing in the lower right corner of $\tau_{\text{BKT}}^2(\t_B; \text{id})$.

\begin{remark}
At first sight, the skew-symmetric matrix $R_{n+1}$ comes from a sequence of artificial manipulations on $\tilde{N}$. We would like to point out another origin for $R_{n+1}$ from where it is more natural to accept it as inhomogeneous coordinate for the associated orthogonal Grassmannian manifold. 
Writing $N_0 = \exp(n_0)$ for $n_0 \in \mathfrak{n}_-$, then we have
\[n_0 = \sum \limits_{i > j \ge 0}r_{i, j}[(-1)^jE_{i, -j} - (-1)^iE_{j, -i}], \qquad \text{with}\quad \left\{\begin{array}{l}
E_{i, j} = \tilde{E}_{i, j} \qquad \text{for } i, j \not\in \{0\}\\
E_{i,0} = \tilde{E}_{i, +0} + \tilde{E}_{i, -0}\quad \text{and} \quad E_{0, -i} = \tilde{E}_{+0, -i} + \tilde{E}_{-0, -i}.
\end{array}\right.\]
Under the spin representation, $n_0$ is represented by the Clifford operator
\[\sum \limits_{i > j \ge 0}r_{i, j}e_{i}e_{j}.\]
From this we see that $R_{n+1} = (r_{i, j})_{0 \le i, j \le n}$ is exactly the inhomogeneous coordinate system for the orthogonal Grassmannian manifold.
\end{remark}

\subsection{Schubert strata of the orthogonal Grassmannian manifolds: the general case}

Now we are ready to reveal the Schubert cell structure of the orthogonal Grassmannian manifold of $\mathfrak{so}_{2n+1}(\mathbb{C})$ and construct general $\tau$-functions for the $B$-Toda and the BKP hierarchies.

Recall that in Remark \ref{rem:Weyliso}, we have identified the normal part of the Weyl groups $\mathfrak{N}_B \cong \mathfrak{N}_D$ for $\mathfrak{so}_{2n+1}$ and $\mathfrak{so}_{2n+2}$, respectively. That is for any strict partition $\lambda = (\lambda_1, \dots, \lambda_k)$ we have $\iota(w_{\lambda}) = \tilde{w}_{\lambda'}$, where $\lambda' = \lambda = (\lambda_1, \lambda_2, \dots, \lambda_{2\ell})$ if $k = 2\ell$ is even and $\lambda' = (\lambda_1, \lambda_2, \dots, \lambda_{2\ell}) = (\lambda_1, \dots, \lambda_k, 0)$ if $k = 2\ell -1$ is odd. In the following, we denote by $w_{\lambda}$ its image $\tilde{w}_{\lambda'}$ under $\iota$ and understand $\lambda$ as the augmented partition $\lambda'$ when the length of $\lambda$ is odd. Then in the identification $\mathfrak{h}^*_D \cong \mathbb{R}^{n+1}$, we have ${w}_{\lambda}(\varepsilon_{\lambda_i}) = -\varepsilon_{\lambda_i} (1 \le i \le 2\ell \le n+1)$ and $w_{\lambda}$ leaves the signs of the other $\varepsilon_j$'s fixed. More concretely, when such a Weyl group element acts from the right (left) on a $(2n+2) \times (2n+2)$ matrix, it transposes the columns (rows) $\lambda_i$ and $-{\lambda}_i$ and keeps the other columns fixed. We denote by $\Omega_{\lambda}(G) = (i_1, i_2, \dots, i_{n+1})$ the $(2n+2) \times (n+1)$ matrix consisting of column $(i_1, i_2, \dots, i_{n+1})$ of the $(2n+2) \times (2n+2)$ matrix $G$, where $i_k = n+1-k$ if $w_{\lambda}(\varepsilon_{i_k}) = \varepsilon_{i_k}$, and $i_k = -(n+1-k)$ otherwise. Denote by $\bar{\Omega}_{\lambda}(G)$ the matrix consisting of all the other columns of $G$.

To proceed, we have the following key lemma.
\begin{lemma}\label{lem:normalformg}
For any $\tilde{N} \in \mathcal{N}_-$, there exists a unique element $P_{\lambda} \in \mathcal{N}_- \cap \mathcal{P}_{w_{\lambda}} = \mathcal{N}_- \cap \dot{w}_{\lambda}\mathcal{P}_0(\dot{w}_{\lambda})^{-1} \in \mathcal{SO}_{2n+1}(\mathbb{C})$ such that other than row $(-0, +{0})$ the $(2n+2) \times (n+1)$ matrix $\bar{\Omega}_{\lambda}(\tilde{N}P_{\lambda})$ is 
in the reduced row echelon form (RREF). 
\end{lemma}
\begin{proof}
An element $P_{\lambda} = (p_{a, b}) \in \mathcal{N}_- \cap \mathcal{P}_{w_{\lambda}}$ has the following matrix form: it is lower unipotent with $p_{i_k, -i_l} = 0$, where $(i_1, i_2, \dots, i_{n+1})$ is the index set associated with $\lambda$ as above. Note that this condition together with $P_{\lambda} \in \mathcal{SO}_{2n+1}$ implies that $p_{+0, -i_l} = p_{-{0}, -i_l} = 0$ and $p_{i_k, +0} = p_{i_k, -{0}} = 0$. For example when $\lambda = \emptyset$ we have $(i_1, i_2, \dots, i_{n+1}) = (n, n-1, \dots, 1, +0)$, then $p_{i_k, -i_l} = 0$ implies that $P_{\emptyset}$ has the form
\[P_{\emptyset} = \begin{pmatrix}[ccc|c|ccc]
 &    & &    &  & & \\
 & A_{\emptyset} & & 0 &  & 0 &\\
 &  &  &   & & & \\
\hline 
& \raisebox{-0.05cm}{$ 0 $} & & I_2 & &\raisebox{-0.05cm}{$ 0 $} & \\
\hline
 &    & &    &  & & \\
 & 0 & & 0 &  & D_{\emptyset}  &\\
 &  &  &   & & & \\
\end{pmatrix},
\]
where $A_{\emptyset}$ and $D_{\emptyset}$ are both $n \times n$ lower unipotent matrices.

Note that there are $\frac{n(n-1)}{2} + |\lambda|$ independent free variables $p_{a, b}$ in $P_{\lambda}$ sitting below both the diagonal and the anti-diagonal, and the other entries in $P_{\lambda}$ are determined from these free variables by the constraint $P_{\lambda} \in \mathcal{SO}_{2n+1}$. To set $\bar{\Omega}_{\lambda}(\tilde{N}P_{\lambda})$ in the RREF other than row $(-0, +{0})$ is the same as to use $P_{\lambda}$ to eliminate all the entries in $\tilde{N}$ at positions where $p_{a, b}$'s are nontrivial, and we only need to do this for elements below both the diagonal and the anti-diagonal since again the other half will follow automatically from the constraint $\tilde{N}P_{\lambda} \in \mathcal{SO}_{2n+1}$. Multiplying $\tilde{N}P_{\lambda}$ out, then we are solving a $(\frac{n(n-1)}{2} + |\lambda|) \times (\frac{n(n-1)}{2} + |\lambda|)$ linear system for $p_{a, b}$. Since $\tilde{N}$ is lower unipotent, we can arrange the equations in this linear system in such a way that the coefficient matrix is lower unipotent, and it follows that for this system we always have a unique solution. We are done.
\end{proof}

As an immediate corollary of the proof, we have
\begin{corollary}
Let $\exp n_{\lambda} = \tilde{N}P_{\lambda}$, then we can write
\[n_{\lambda} = \sum \limits_{i_k > |i_l|}r_{i_k, i_l}[(-1)^{i_l}E_{i_k, -i_l} - (-1)^{i_k}E_{i_l, -i_k}],\]
where $(i_1, i_2, \dots, i_{n+1})$ is the indix associated with the strict partition $\lambda$ as before and again the following convention is used
\[\left\{\begin{array}{l}
E_{i, j} = \tilde{E}_{i, j} \qquad \text{for } i, j \not\in \{+0, -{0}\}\\
E_{i,0} = \tilde{E}_{i, +0} + \tilde{E}_{i, -0}\quad \text{and} \quad E_{0, -i}= \tilde{E}_{+0, -i} + \tilde{E}_{-0, -i}.
\end{array}\right.\]
\end{corollary}

\begin{remark}
\begin{enumerate}
\item The index $i_k$ such that $i_k > |i_l|$ for some $i_l$ belongs to the set 
\[\lambda^{\perp} := \{1, 2, \dots, n\} \backslash \{\lambda_1, \lambda_2, \dots, \lambda_{2\ell}\}\]
\item 
Note that there are $(\frac{n(n+1)}{2} - |\lambda|)$ many non-trivial $r_{i, j}$'s and they consist of the inhomogeneous coordinate system for the Schubert cell $\mathcal{L}_{\lambda}$ determined by $w_{\lambda}$.
\item 
We also have that
\begin{align*}
\tilde{N}P_{\lambda} 
& = I + \sum \limits_{i_k > |i_l|}r_{i_k, i_l}[(-1)^{i_l}E_{i_k, -i_l} - (-1)^{i_k}E_{i_l, -i_k}] + \frac{1}{2}\left\{\sum \limits_{i_k > |i_l|}r_{i_k, i_l}[(-1)^{i_l}E_{i_k, -i_l} - (-1)^{i_k}E_{i_l, -i_k}]\right\}^2\\
& = I + \sum \limits_{i_k > |i_l|}r_{i_k, i_l}[(-1)^{i_l}E_{i_k, -i_l} - (-1)^{i_k}E_{i_l, -i_k}] + \frac{1}{2}\left\{\sum \limits_{i_k > 0}r_{i_k, 0}[E_{i_k, 0} - (-1)^{i_k}E_{0, -i_k}]\right\}^2\\
& = I + \sum \limits_{i_k > |i_l|}r_{i_k, i_l}[(-1)^{i_l}E_{i_k, -i_l} - (-1)^{i_k}E_{i_l, -i_k}] + \sum \limits_{i_k, i_l > 0}(-1)^{i_l + 1}r_{i_k, 0}r_{i_l, 0}E_{i_k, -i_l}.
\end{align*}
Note that
\[\tilde{M} = \sum \limits_{k = 0}^n (-1)^k[\tilde{E}_{+k, -k} + \tilde{E}_{-k, +k}]\]
then the matrix $\tilde{N}P_{\lambda}\tilde{M}$ has the following appealing form
\[\tilde{N}P_{\lambda}\tilde{M} = \tilde{M} + \sum \limits_{i_k > |i_l|}r_{i_k, i_l}[E_{i_k, i_l} - E_{i_l, i_k}] - \sum \limits_{i_k, i_l > 0}r_{i_k, 0}r_{i_l, 0}E_{i_k, i_l}.\]
\end{enumerate}
\end{remark}

In summary, we have the following canonical form for points in the orthogonal Grassmannian manifold of $\mathfrak{so}_{2n+1}$ realized in $(\tilde{V}, \tilde{Q})$.
\begin{theorem}\label{thm:normalformL}
Let $(i_1, i_2, \dots, i_{n+1})$ be the index set associated with the strict partition $\lambda$, then an element belong to the Schubert cell $\mathcal{L}_{\lambda}$ of the orthogonal Grassmannian manifold corresponding to the Weyl group element $w_{\lambda} \in \mathfrak{N}_B$ could take column $(i_1, i_2, \dots, i_{n+1})$ of the following matrix as a representative
\[L_{\lambda} := \tilde{N}P_{\lambda}\tilde{M} = \tilde{M} + \sum \limits_{i_k > |i_l|}r_{i_k, i_l}[E_{i_k, i_l} - E_{i_l, i_k}] - \sum \limits_{i_k, i_l > 0}r_{i_k, 0}r_{i_l, 0}E_{i_k, i_l},\]
where $r_{i_k, i_l} \in \mathbb{C} \ (i_k > |i_l|)$ are arbitrary constants which provide a system of inhomogeneous coordinates for $\mathcal{L}_{\lambda}$. 
\end{theorem}

\subsection{$\tau$-functions for the $B$-Toda and the BKP hierarchy: the general case}

With proper column operation on $L_{\lambda}$, we have
\begin{lemma}\label{lem:ReducedLagrangian}
Viewed as a subspace of $(\tilde{V}, \tilde{Q})$, representative of elements in the Schubert cell $\mathcal{L}_{\lambda}$ can be further taken as ${\Omega}_{\lambda}(\tilde{L}_{\lambda})$, with
\[\tilde{L}_{\lambda} := \tilde{N}P_{\lambda}\tilde{M}\tilde{U}_{\lambda} = \tilde{M} + \sum \limits_{i_k > |i_l|}r_{i_k, i_l}[\tilde{E}_{i_k, i_l} - \tilde{E}_{i_l, i_k}],\]
where
\[\tilde{U}_{\lambda} = I + \sum \limits_{i_k > 0}r_{i_k, 0}(\tilde{E}_{\delta, i_k} - (-1)^{i_k}\tilde{E}_{-i_k, -\delta}),\]
where $\delta = +0$ if the length of the non-augmented $\lambda$ is even, and $\delta = -0$ if otherwise.
\end{lemma}

\begin{example}
Taking $\tilde{N} = \exp(\Theta_f(\t_B))$, then the canonical form of $\tilde{N}$ which lies in the big cell of the orthogonal Grassmannian manifold is
\begin{align*}
\tilde{\mathcal{P}}_{2n+2} & = {\Omega}_{\emptyset}(\exp(\Theta_{f}(\mathbf{t}_B))P_{\emptyset}\tilde{M}\tilde{U}_{\text{id}})\\
& = \begin{pmatrix}[ccccc]
0 & 0 & 0 & \cdots & (-1)^n  \\
\vdots & \vdots & \vdots & \udots & \vdots  \\
0 & 0 & 1 & \cdots & \cdots  \\
0 & -1 & 0 & \cdots & \vdots  \\
1 & 0 & 0 & \cdots & 0 \\
0 & -Q_{1,0} & -Q_{2,0} & \cdots & -Q_{n,0} \\
Q_{1,0} & 0 & -Q_{2, 1}  & \cdots & -Q_{n, 1} \\
Q_{2,0} & Q_{2,1}  & 0 & \cdots & -Q_{n, 2}\\
\vdots & \vdots & \vdots & \ddots & \vdots\\
Q_{n,0} & Q_{n, 1}  & Q_{n, 2} & \cdots & 0 
\end{pmatrix}.
\end{align*}
\end{example}

\begin{theorem}\label{thm:generaltaufunction}
$\tau_{\text{BKT}}(\t_B; w_{\lambda})$-function for the $B$-Toda is the Pfaffian of the following skew-symmetric matrix
\begin{align*}
W_{\lambda} = & \sum \limits_{j > i \ge 0}Q_{ji}(\t_B) (\tilde{E}_{-i, -j} - \tilde{E}_{-j, -i}) + \sum \limits_{i_k > |i_l|}\left[(\tilde{E}_{-i_k, i_k} - \tilde{E}_{i_k, -i_k}) + (-1)^{\text{min}\{0, i_{l}\}}r_{i_k, i_l}(\tilde{E}_{i_k, i_l} - \tilde{E}_{i_l, i_k})\right] + \\
& \qquad + \sum \limits_{j = 1}^{\ell}(\tilde{E}_{\lambda_{2j}, \lambda_{2j-1}} - \tilde{E}_{\lambda_{2j-1}, \lambda_{2j}}) + (1 - \delta_{\lambda_{2\ell, 0}})(\tilde{E}_{-0, +0} - \tilde{E}_{+0, -0}),
\end{align*}
where $(i_1, i_2, \dots, i_{n+1})$ is the index associated with the augmented strict partition $\lambda = (\lambda_1, \lambda_2, \dots, \lambda_{2\ell})$ with $\lambda_{2\ell} = 0$ if the length of the non-augmented $\lambda$ is odd.
\end{theorem}

\begin{proof}
Let $A = \tilde{\mathcal{P}}_{2n+2}$ and $B = {\Omega}_{\lambda}(\tilde{L}_{\lambda})$, then
\[\tau_{\text{BKT}}^2(\t_B; w_{\lambda}) = \det(A^TB).
\]
The desired result can be obtained by manipulating the following equality
\[
\det(A^TB) = (-1)^{p(q+1)}\det\begin{pmatrix}
A^T & 0 \\
-I_p & B 
\end{pmatrix},
\]
where $A, B$ are matrices of size $p \times q$. 
\end{proof}

\begin{remark}
Theorem \ref{thm:taufunctiongeneric}, Theorem \ref{thm:normalforminf} and Theorem \ref{thm:taufunctioninf} can be easily obtained from the relevant results for $B$-Toda by taking the limit $n \to \infty$.
\end{remark}

From Theorem \ref{thm:normalformL} and Theorem \ref{thm:generaltaufunction}, we immediately deduce that
\begin{corollary}\label{cor:generaltaufunction}
The $\tau$-function of the BKP hierarchy in the Schubert cell $\mathcal{L}_{\lambda}$ associated with the Weyl group element $w_{\lambda}$ always has the following Schur-Q expansion
\[\tau(\t_B; w_{\lambda}) = Q_{\lambda} + \sum \limits_{\lambda \varsubsetneqq \mu \in DP}R_{\mu}Q_{\mu},\]
where $Q_{\nu}$ is the Schur-$Q$ function associated with the strict partition $\nu$ and $R_{\nu}$ is the Pfaffian of certain minors of the skew-symmetric matrix in Theorem \ref{thm:generaltaufunction}. That is,  $Q_{\lambda}$ is always the leading term in the Schur-$Q$ expansion of $\tau(\t_B; w_{\lambda})$.
\end{corollary}

\subsection{Pure spinors of E.Cartan}\label{sec:purespinor}
Theorem \ref{thm:generaltaufunction} also provides us a constructive description for the pure spinors of E. Cartan (see Section \ref{sec:Orthogonal}) which we now explain. Recall that in Section \ref{sec:f-KTAB} and Section \ref{sec:ABKP} we saw that under the Boson-Fermion correspondence geometrically $B$-Toda and BKP describes the group orbit $\mathcal{G} \cdot v^{\omega_1}$, where $\mathcal{G}$ is a finite-dimensional or a properly defined infinite-dimensional orthogonal group and $v^{\omega_1}$ is the highest weight vector in the corresponding spin module. The Bonson-Fermion correspondence $\sigma_0^B: \mathcal{G} \cdot v^{\omega_1} \cong \mathbb{C}[t_1, t_3, \dots]$ in type $B$ is given explicitly by (c.f. \cite{Date-Jimbo-Kashiwara-Miwa1982, Jimbo-Miwa1983, You1989})
\begin{equation}\label{eq:BonsonFermionBPoly}
\sigma^B_0(w_{\lambda} \cdot v^{\omega_1}) = Q_{\lambda}(\t_B),
\end{equation}
where $l(\lambda)$ is the length of the partition $\lambda$. Note that $(Q_{\lambda}(\t_B))_{\lambda \in \text{DP}}$ is a basis for $\mathbb{C}[t_1, t_3, \dots]$. Combining Equation \eqref{eq:BonsonFermionBPoly} and Theorem \ref{thm:normalformL}, Theorem \ref{thm:generaltaufunction} and Corollary \ref{cor:generaltaufunction}, we obtain the proof for Theorem \ref{thm:purespinor}.


\subsection{An example}
We use the following simple example to illustrate the content of this section. Taking $n = 2$ and $w = w_1$, we then have
\[L_{1} = \begin{pmatrix}
& & & & & 1\\
& & & & -1 & -r_{2, -1}\\
& & & 1 & 0 & -r_{2, 0} \\
& & 1 & 0 & 0 & -r_{2, 0} \\
& -1 & 0 & 0 & 0 & 0 \\
1 & r_{2, -1} & r_{2, 0} & r_{2, 0} & 0 & -r_{2, 0}^2
\end{pmatrix}, \qquad \tilde{L}_{1} = \begin{pmatrix}
& & & & & 1\\
& & & & -1 & -r_{2, -1}\\
& & & 1 & 0 & -r_{2, 0} \\
& & 1 & 0 & 0 & 0 \\
& -1 & 0 & 0 & 0 & 0 \\
1 & r_{2, -1} & r_{2, 0} & 0 & 0 & 0
\end{pmatrix}\]
and
\begin{align*}
\tau_{\text{BKT}}^2(w_1, \t_B)  & = \det \begin{pmatrix}
0 & -Q_{2,1} & -Q_{2,0} & 0 & 0 & 1 & 0 & 0 & 0\\
Q_{2,1} & 0 & -Q_{1,0} & 0 & -1 & 0 & 0 & 0 & 0\\
Q_{2,0} & Q_{1,0} & 0 & 1 & 0 & 0 & 0 & 0 & 0 \\
0 & 0 & 0 & 0 & 0 & -1 & 0 & 0 & 1\\
0 & 0 & 0 & 0 & -1 & 0 & 0 & 0 & -r_{2, 1}\\
0 & 0 & 0 & -1 & 0 & 0 & 0 & 0 & -r_{2, 0}\\
0 & 0 & -1 & 0 & 0 & 0 & 0 & 1 & 0\\
0 & -1 & 0 & 0 & 0 & 0 & -1 & 0 & 0\\
-1 & 0 & 0 & 0 & 0 & 0 & r_{2, 1} & r_{2, 0} & 0
\end{pmatrix}\\
 & = \det \begin{pmatrix}
 0 & -Q_{2,1} & -Q_{2,0} & 0 & 0 & 1 \\
 Q_{2,1} & 0 & -Q_{1,0} & 0 & 0 & r_{2, -1} \\
 Q_{2,0} & Q_{1,0} & 0 & 0 & 0 & -r_{2, 0} \\
 0 & 0 & 0 & 0 & 1 & 0 \\
 0 & 0 & 0 & -1 & 0 & 0 \\
 -1 & -r_{2, -1} & r_{2, 0} & 0 & 0 & 0
\end{pmatrix}.
\end{align*}

\section{An application: KdV as $4$-reduction of BKP}\label{sec:KdVBKP}
In \cite{Alexandrov2021}, A. Alexandrov proved that solutions of KdV satisfy BKP hierarchy up to rescaling of the time parameters by $2$. More precisely, we have
\begin{theorem}[\cite{Alexandrov2021}]\label{thm:KdVinBKPb}
For any KdV $\tau$-function,
\begin{equation}\label{eq:KdVinBKP}
\tau(\t_B) = \tau_{\text{KdV}}(\t_B \slash 2)
\end{equation}
is a $\tau$-function of the BKP hierarchy.
\end{theorem}
Alexandrov showed this result by directly comparing the Hirota bilinear identities the $\tau$-functions of KdV and the $\tau$-functions of BKP should satisfy and did not provide an explanation from the first principle why such ``surprising'' result should be true. Later in \cite{vandeLeur2021}, J. van de Leur tried to give Theorem \ref{thm:KdVinBKPb} a representation theoretical interpretation, and succeeded in the case of rational solutions. 

In the following, we show that KdV is nothing but the $4$-reduction of BKP hierarchy and Theorem \ref{thm:KdVinBKPb} follows from the lower rank coincidence of Kac-Moody Lie algebras $A_1^{(1)} \cong D_2^{(2)}$. The other coincidences of lower rank Kac-Moody Lie algebras and their implications to integrable systems will be investigated elsewhere. For example, reading from the Dynkin diagrams (c.f. \cite{Kac1990}), we may deduce that KdV is also a $2$-reduction of CKP hierarchy.

\subsection{Realization of the Kac-Moody Lie algebras $A^{(1)}_1$ and $D^{(2)}_2$}\label{sec:A11andD22}
First we note that $A^{(1)}_1$ and $D_2^{(2)}$ share the same Dynkin diagram (c.f. \cite{Kac1990}, page 54-55) thus the same generalized Cartan matrix $A$: 
\[\circ \Leftrightarrow \circ, \qquad A = \begin{pmatrix} 2 & -2 \\ -2 & 2 \end{pmatrix}.\]
And from the general theory they should give isomorphic Kac-Moody Lie algebras, that is, they should give different realization of the abstract Lie algebra generated by six elements $h_0, h_1, e_0, e_1, f_0, f_1$ subjected to the following relations
\begin{align}\label{eq:SerreA11}
& [h_0, h_1] = 0, \nonumber\\
& [e_i, f_j] = \delta_{ij}h_i,\nonumber\\
& [h_i, e_j] = A_{ij}e_j,\\
& [h_i, f_j] = -A_{ij}f_j,\nonumber\\
& [e_i, [e_i, [e_i, e_j]]] = 0 = [f_i, [f_i, [f_i, f_j]]] \qquad \text{if} \quad i \ne j. \nonumber
\end{align}

We will construct the isomorphism $A_1^{(1)} \cong D_2^{(2)}$ and the corresponding $\tau$-functions explicitly by embedding both of them into $\mathfrak{a}_{\infty}$.

Recall that $A_{\ell}^{(1)}$ can be realized as a central extension of loop algebras (c.f. \cite{Kac1990}). 

The loop algebra $\widetilde{\mathfrak{gl}}_{\ell + 1}$ is defined as $\mathfrak{gl}_{\ell + 1}(\mathbb{C}[z, z^{-1}])$, i.e. as the complex Lie algebra of $(\ell + 1) \times (\ell + 1)$ matrices with Laurent polynomials as entries. A basis of $\widetilde{\mathfrak{gl}}_{\ell + 1}$ can be chosen as
\[{E}_{i, j}(m) \equiv {E}_{i, j} \otimes z^m, \qquad 1 \le i, j \le \ell + 1 \quad \text{and} \quad m \in \mathbb{Z},\]
with the following commutation relations
\[[{E}_{i, j}(m), {E}_{k, l}(n)] = \delta_{jk}{E}_{i, l}(m + n) - \delta_{li}{E}_{k, j}(m + n).\]
As $\mathfrak{gl}_{\ell+1}$ naturally acts on $\mathbb{C}^{\ell + 1}$, which has a standard basis ${e}_1, \dots, e_{\ell +1}$ of $\ell + 1$ column vectors in which $e_{i} (1 \le i \le \ell + 1)$ has $1$ in the $i$-th row and $0$ elsewhere. The loop algebra $\widetilde{\mathfrak{gl}}_{\ell + 1}$ acts on $\mathbb{C}[z, z^{-1}]^{\ell + 1}$, which consists of $(\ell + 1) \times 1$ column vectors with Laurent polynomials in $z$ as entries. The vectors
\[e_{k(\ell + 1)+j} = {e}_j \otimes z^{-k}\]
form a basis of $\mathbb{C}[z, z^{-1}]^{\ell + 1}$ (over $\mathbb{C}$) indexed by $\mathbb{Z}$. Thus we obtain an identification of $\mathbb{C}[z, z^{-1}]^{\ell + 1}$ with $\mathbb{C}^{\infty}$. Note that
\[{E}_{i, j}(k)e_{s(\ell + 1)+j} = e_{(\ell + 1)(s-k) + i},\]
which induce the following embedding of Lie algebras
\[\begin{array}{rcl}
\iota: \tilde{\mathfrak{gl}}_{\ell + 1} & \hookrightarrow & \overline{\mathfrak{a}}_{\infty}\\
{E}_{i, j}(k) & \mapsto & \sum \limits_{s \in \mathbb{Z}}E_{(\ell + 1)(s-k)+i, (\ell + 1)s + j}.
\end{array}
\]
An element of $\widetilde{\mathfrak{gl}}_{\ell + 1}$ has the form
\[a(z) = \sum \limits_{k}a_k \otimes z^k \qquad (a_k \in \mathfrak{gl}_{\ell + 1}),\]
where $k$ runs over a finite subset of $\mathbb{Z}$. Then the infinite matrix $\iota(a(z)) \in \overline{\mathfrak{a}}_{\infty}$ takes the form
\[\iota(a(z)) = \begin{pmatrix}
  \ddots & \ddots & \ddots & \cdots & \cdots & \cdots & \cdots\\
 \cdots & a_{-1} & a_0 & a_1 & \cdots & \cdots & \cdots\\
 \cdots & \cdots & a_{-1} & a_0 & a_1 & \cdots & \cdots \\
 \cdots & \cdots & \cdots & a_{-1} & a_0 & a_1 & \cdots \\
 \cdots & \cdots & \cdots & \cdots & \ddots & \ddots & \ddots
\end{pmatrix}.\]

Thus the central extension of $\widetilde{\mathfrak{gl}}_{\ell + 1}$
\[\widehat{\mathfrak{gl}}_{\ell + 1} = \widetilde{\mathfrak{gl}}_{\ell + 1} \oplus \mathbb{C}c,\]
will have a linear representation as a sub-Lie algebra of $\mathfrak{\mathfrak{a}}_{\infty}$, the central extension of $\overline{\mathfrak{a}}_{\infty}$. The restriction of the two-cycle $\alpha$ defined in \eqref{eq:twocycle} on $\mathfrak{a}_{\infty}$ now takes the following form
\[\alpha(\iota({E}_{ij}(m)), \iota({E}_{kl}(n))) = \delta_{il}\delta_{jk}\delta_{m+n, 0}m.\]
It follows by linearity that if $X(m) = X \otimes z^m, Y(m) = Y \otimes z^n$, then
\[\alpha(\iota(X(m)), \iota(Y(n))) = \delta_{m+n, 0}m\text{tr}(XY),\]
where tr denotes the trace in $\mathfrak{gl}_{\ell + 1}$. For general elements $a(z), b(z)$ in $\widetilde{\mathfrak{gl}}_{\ell + 1}$, we have
\[\alpha(\iota(a(z)), \iota(b(z))) = \text{Res}_0\text{tr}(a'(z)b(z)),\]
where $a'(z)$ is the derivative of $a$ with respect to $z$ and Res$_0$ is the residue at $z = 0$, i.e., the coefficient of $1 \slash z$.

\begin{example}
Taking $\ell = 1$, we have the following classical realization of $A_1^{(1)}$ as $\widehat{\mathfrak{sl}}_2$ with Chevalley basis:
\begin{align*}
& e_1 = \begin{pmatrix} 0 & 1 \\ 0 & 0 \end{pmatrix} \otimes 1, \quad f_1 = \begin{pmatrix} 0 & 0 \\ 1 & 0 \end{pmatrix} \otimes 1, \quad h_1 = [e_1, f_1] = \begin{pmatrix} 1 & 0 \\ 0 & -1 \end{pmatrix} \otimes 1,\\
& e_0 = f_1 \otimes z, \quad f_0 = e_1 \otimes z^{-1}, \quad h_0 = [e_0, f_0] = -h_1 + c.
 \end{align*}
 It is easy to check that these elements indeed satisfy all the commutation relations in \eqref{eq:SerreA11}. 
\end{example}

\begin{example}
$D_2^{(2)}$ can be realized in $\mathfrak{a}_{\infty}$ as the intersection of $\iota(\widehat{\mathfrak{sl}}_4)$ and $\mathfrak{b}_{\infty}$ (c.f. \cite{Jimbo-Miwa1983}, page 977-978). As a sub-Lie algebra of $\widehat{\mathfrak{sl}}_4$, it has Chevalley basis:
\begin{align}\label{eq:BasisD22}
& e_1 = \begin{pmatrix} 0 & 0 & 0 & 0\\ 0 & 0 & 1 & 0\\ 0 & 0 & 0 & 1\\ 0 & 0 & 0 & 0 \end{pmatrix} \otimes 1, \quad f_1 = \begin{pmatrix} 0 & 0 & 0 & 0\\ 0 & 0 & 0 & 0\\ 0 & 2 & 0 & 0\\ 0 & 0 & 2 & 0 \end{pmatrix} \otimes 1, \nonumber\\
& e_0 = \begin{pmatrix} 0 & 1 & 0 & 0 \\ 0 & 0 & 0 & 0\\ 0 & 0 & 0 & 0 \\ 0 & 0 & 0 & 0 \end{pmatrix} \otimes 1 + \begin{pmatrix} 0 & 0 & 0 & 0 \\ 0 & 0 & 0 & 0 \\ 0 & 0 & 0 & 0\\ 1 & 0 & 0 & 0 \end{pmatrix} \otimes {z}, \quad f_0 = \begin{pmatrix} 0 & 0 & 0 & 0 \\ 2 & 0 & 0 & 0 \\ 0 & 0 & 0 & 0\\ 0 & 0 & 0 & 0\end{pmatrix} + \begin{pmatrix}0 & 0 & 0 & 2\\ 0 & 0 & 0 & 0\\ 0 & 0 & 0 & 0\\ 0 & 0 & 0 & 0 \end{pmatrix} \otimes {z}^{-1}\\
& h_1 = \begin{pmatrix} 0 & 0 & 0 & 0 \\ 0 & 2 & 0 & 0\\ 0 & 0 & 0 & 0 \\ 0 & 0 & 0 & -2 \end{pmatrix} \otimes 1,  \quad h_0 = [e_0, f_0] = -h_1 + 2c. \nonumber
\end{align}
 It can be verified that the above six elements also satisfy all the Serre relations in \eqref{eq:SerreA11}, so it gives another presentation for $A^{(1)}_1$. We also note that these elements generate exactly all the elements in $\widehat{\mathfrak{sl}}_4 \cap \mathfrak{b}_{\infty}$. To see this clearly we just need to renumber the rows and columns of their image $\iota(X)$ in $\mathfrak{a}_{\infty}$ by shifting them by $3$, i.e. we have
\begin{align}\label{eq:D22inB}
& \iota(e_0) = \sum \limits_{i \in \mathbb{Z}} (E_{4i+1, 4i+2} + E_{-4i-2, -4i-1}), \qquad \iota(f_0) = 2 \sum \limits_{i \in \mathbb{Z}}(E_{4i+2, 4i+1} + E_{-4i-1, -4i-2}), \nonumber\\
& \iota(e_1) = \sum \limits_{i \in \mathbb{Z}}(E_{4i, 4i+1} + E_{-4i-1, -4i}), \qquad \iota(f_1) = 2 \sum \limits_{i \in \mathbb{Z}}(E_{4i+1, 4i} + E_{-4i, -4i-1}),\\
& \iota(h_1) = 2 \sum \limits_{i \in \mathbb{Z}} (E_{4i-1, 4i-1} - E_{4i+1, 4i+1}). \nonumber
\end{align}
\end{example}

\begin{remark}
As noted in Section \ref{sec:ABKP} the shifts $\nu_s$ induce isomorphic irreducible representations and give the same $\tau$-functions for $\mathfrak{a}_{\infty}$.
\end{remark}

\begin{remark}\label{rem:subHeisenberg}
We also note that $D_2^{(2)} \subset \mathfrak{a}_{\infty}$ contains the Heisenberg algebra $\mathcal{H}_{\text{KdV}} := \{H_{2k+1}, k \in \mathbb{Z}; c\}$ as a sub-Lie algebra. For example, we have
\begin{align*}
& H_1 = \iota(e_0 + e_1), \quad H_3 = -\text{ad}_{[\iota(e_0), \iota(e_1)]}(\iota(e_0-e_1)), \quad H_5 = -\text{ad}^2_{[\iota(e_0), \iota(e_1)]}(\iota(e_0+e_1)), \dots.
\end{align*}
\end{remark}

\subsection{KdV as $4$-reduction of BKP}
Realizing $A_1^{(1)}$ as a sub-Lie algebra $\iota(\widehat{\mathfrak{sl}}_2) \subset \mathfrak{a}_{\infty}$ of $\mathfrak{a}_{\infty}$, the $\tau$-functions of KdV are obtained from the $\tau$-functions of the KP hierarchy by restricting to the group orbit of $\widehat{\mathcal{SL}}_2$. More explicitly, let
\[{H}_{\text{KdV}}(t_1, t_3, t_5\dots) = \sum \limits_{n > 0}t_{2n-1}H_{2n-1}, \quad \text{and} \quad \bar{H}_{\text{KdV}}(t_1, t_3, t_5\dots) = \sum \limits_{n > 0}t_{2n-1}H_{-2n+1},\]
then $\tau$-function of the KdV hierarchy has the following form
\begin{equation}\label{eq:KdVA}
\tau^A_{\text{KdV}}(t_1, t_3, t_5, \dots; g) = \langle \bar{R}(\exp(\bar{H}_{\text{KdV}}(\t_B)))v_0, \bar{R}(g) \cdot v_0 \rangle, \qquad g \in \widehat{\mathcal{SL}}_2.
\end{equation}
Recall that for the finite-dimensional simple Lie algebras, the lower rank coincidence $\mathfrak{sl}_2 \cong \mathfrak{so}_3$ tells us that the standard representation of $\mathfrak{sl}_2$ on $\mathbb{C}^2$ viewed as a representation of $\mathfrak{so}_3$ is exactly the spin representation of $\mathfrak{so}_3$ as an orthogonal Lie algebra (c.f. \cite{Fulton-Harris2013}). Here we have a similar situation for  infinite-dimensional Lie algebras: the $\tau$-function of KdV given in \eqref{eq:KdVA} is actually associated with the spin representation of $D_2^{(2)}$ as a sub-Lie algebra of $\mathfrak{b}_{\infty}$. So to get the same $\tau$-function from the presentation $D_2^{(2)} = \widehat{\mathfrak{sl}}_4 \cap \mathfrak{b}_{\infty} \subset \mathfrak{b}_{\infty}$ of $A_1^{(1)}$, we need to restrict the $\tau$-function associated with the spin representation of $\mathfrak{b}_{\infty}$ to the group orbit of $D_2^{(2)}$. 

Now the rescaling of the time parameters by $2$ in Theorem \ref{thm:KdVinBKPb} can also be easily explained. For $g$ in the group of $D_2^{(2)}$ there exists $\tilde{g} \in \widehat{\mathcal{SL}}_2$ such that
\begin{align*}
\tau^B_{\text{KdV}}(t_1, t_3, t_5, \dots; g) & = \langle v^{\omega_1}, \bar{\rho}(\exp(H_{\text{KdV}}(t_1, t_3, \dots))g)v^{\omega_1}\rangle\\
& = \langle \bar{\rho}(\exp(\bar{H}_{\text{KdV}}({t_1}, {t_3}, {t_5}\dots))) v^{\omega_1}, \bar{\rho}(g)v^{\omega_1}\rangle\\
& = \langle \bar{R}(\exp(\bar{H}_{\text{KdV}}(\frac{t_1}{2}, \frac{t_3}{2}, \frac{t_5}{2}\dots)))v_0, \bar{R}(\tilde{g}) \cdot v_0 \rangle\\
& = \tau^A_{\text{KdV}}(\frac{t_1}{2}, \frac{t_3}{2}, \frac{t_5}{2}\dots; \tilde{g}),
\end{align*}
where $(\bar{\rho}, \bar{S})$ is the spin representation on certain completion $\bar{S}$ of $S$ with highest weight vector $v^{\omega_1}$, and the third equality is obtained by comparing \eqref{eq:Bonson-FermionA} and \eqref{eq:Bonson-FermionB}. Since $\tau^B_{\text{KdV}}$ is the $\tau$-function of $\mathfrak{b}_{\infty}$ restricted to the group orbit of $D_2^{(2)}$, it is a $\tau$-function for the BKP hierarchy, and we obtain another proof of Theorem \ref{thm:KdVinBKPb}.

\begin{example}
For $n \in \mathbb{N}$, and partition $\lambda = (n, n-1, \dots, 1)$ which is associated with a Weyl group element of $A_1^{(1)}$ (c.f. \cite{Jimbo-Miwa1983}), we have
\[S_{\lambda}(\t_B \slash 2) = 2^{-n \slash 2}Q_{\lambda}(\t_B).\]
\end{example}

Now we can completely characterize the KdV orbits inside the BKP hierarchy. Note that the relation $D_2^{(2)} = \widehat{\mathfrak{sl}}_4 \cap \mathfrak{b}_{\infty}$ tells us that KdV is nothing but the $4$-reduction of $\mathfrak{b}_{\infty}$. More precisely, consider the sub-Lie algebra $\mathfrak{b}^{(4)}_{\infty}$ of $\mathfrak{b}_{\infty}$ consisting of those elements whose adjoint representations commute with $\nu_4(a_{i, j}) = a_{i-4, j-4}$:
\[\mathfrak{b}^{(4)}_{\infty} = \{X \in \mathfrak{b}_{\infty}\ |\ [\text{ad}\ X, \nu_4] = 0\}.\]
Then $\mathfrak{b}^{(4)}_{\infty}$ contains a Heisenberg subalgebra $\mathcal{H}_4$ spanned by $H_{4k}$ and $1$, and split into the direct sum of $\bigoplus \limits_{k \in \mathbb{Z}}\mathbb{C}H_{4k}$ and an algebra isomorphic to $D_2^{(2)}$. Realizing the Lie group of $D_2^{(2)}$ inside $\mathcal{B}_{\infty}$ is also easy, the Lie group of $\mathfrak{b}^{(4)}_{\infty}$ is the direct product of the Lie group of $\mathcal{H}_4$ and the Lie group of $D_2^{(2)}$.

Since elements $\mathcal{H}_4$ and elements in $D_2^{(2)}$ mutually commute with each other, a $D_2^{(2)}$ module  can be extended to a $\mathfrak{b}^{(4)}_{\infty}$ module such that $\mathcal{H}_{4k}, k \ne 0$ trivially acts on it. 
The group orbit of $D_2^{(2)} \subset \mathfrak{b}_{\infty}$ on the spin module $\bar{S}$ is characterized by the following conditions
\[\bar{S}^{(4)} := \{v \in \bar{\rho}(\mathcal{B}_{\infty}) \cdot v^{\omega_1}\ |\ \bar{\rho}(H_{4k}) \cdot v = 0, \quad k = 1, 2, \dots\}.\]
We note that this subspace of $\bar{S}$ is irreducible and invariant under the action of $D_2^{(2)}$. We can also characterize this orbit by $\tau$-functions through Theorem \ref{thm:ABKPrelation} as follows. Let $\bar{F}_B \subset \bar{F}^{(0)}$ be the irreducible $\mathfrak{b}_{\infty}$-submodule of $\bar{F}^{(0)}$ generated by $v_0$, then again the group orbit of $D_2^{(2)}$ is characterized by
\[\bar{F}^{(4)}_B := \{v \in \bar{F}_B \ |\ \bar{r}(H_{4k}) \cdot v = 0, \quad k = 1, 2, \dots\}.\]
There is a natural projection $\pi_B: \mathcal{B}_{\infty} \to \text{OG}(\bar{V}) = \mathcal{B}_{\infty} \slash \bar{\mathcal{P}}_0$ where $\bar{\mathcal{P}}_0 = \{\bar{T} \in \mathcal{B}_{\infty} | T\bar{L}_0 = \bar{L}_0\}$ is the isotropic subgroup of the Lagrangian subspace $\bar{L}_0 \subset \bar{V}$ spanned by $e_{i} (i < 0)$. The Lagrangian subspace associated with $D_2^{(2)}$ are all the Lagrangian subspaces of $\mathfrak{b}_{\infty}$ such that $\nu_4(W) \subset W$, where $\nu_4(e_i) = e_{i-4}$. We note that this characterization of group orbit of KdV is the exact analogue of the well-known fact that KdV is a $2$-reduction of KP (c.f. \cite{Segal-Wilson1985}).

At last we also describe a construction of the module $\bar{F}^{(4)}_B$ as induction from an irreducible module of a sub-Lie algebra of $D_2^{(2)}$ instead of as restriction from $\bar{F}_B$.  By Remark \ref{rem:subHeisenberg}, $\iota(D_2^{(2)}) \subset \mathfrak{a}_{\infty}$ contains $\mathcal{H}_{\text{KdV}}$ as a Heisenberg subalgebra, then following the proof of Theorem 4.7 in \cite{Lepowsky-Wilson1978} word by word, we see that $\bar{F}^{(4)}_B$ is the $\mathcal{H}_{\text{KdV}}$-module generated by $v_0$. 

From any description provided above, we immediately obtain the following theorem regarding solutions of the KdV hierarchy. 
\begin{theorem}
For any sequence of complex numbers $\mathbf{a}=(a_1, a_3, a_5, \dots)$, the following formal series is a $\tau$-function for the KdV hierarchy
\[\tau_{\text{KdV}}({t_1}, {t_3}, \dots) = \sum \limits_{\lambda \in \text{DP}}Q_{\lambda}(\mathbf{a})Q_{\lambda}(2\t_B).\]
\end{theorem}
We note that a very special case of Schur-Q expansion similar to but not exact the same as this one which is of great importance in Gromov-Witten theory has already been obtained in \cite{Alexandrov2021c} and \cite{Liu-Yang2022} independently and the authors used Schur-Q expansion of a class of the so-called hypergeometric $\tau$-functions of BKP hierarchy (c.f. \cite{Orlov2003}).


\raggedright



\bibliographystyle{plain}
\bibliography{BKP.bib}
\end{document}